\documentclass[12pt]{article}

\usepackage{./qc}
\usepackage{authblk}
\usepackage[margin=1in]{geometry}
\usepackage{xspace}
\usepackage{tcolorbox}
\usepackage{subcaption}
\usepackage{stmaryrd}
\usepackage{resizegather}
\usepackage{array, multirow}
\usepackage{adjustbox}
\usepackage{comment}
\usepackage{mathdots}
\usepackage{expl3}
\usepackage{aligned-overset}
\usepackage{enumitem}
\usepackage{resizegather}

\usepackage[backend=biber, style=alphabetic, backref=true, hyperref=true, maxbibnames=99]{biblatex}
\usepackage{imakeidx}
\makeindex[title={\protect\hypertarget{theidx}{Index}}]

\usepackage{tikz}
\usetikzlibrary{positioning}
\usetikzlibrary{quantikz2}

\usepackage{./quantikzplus}

\usetikzlibrary{external}
\ExplSyntaxOn
\tl_new:N  \l__autotikz_code_tl
\tl_new:N  \l__autotikz_hash_tl

\NewDocumentCommand \autoname { m }
 {
   \tl_set:Nn \l__autotikz_code_tl { #1 }

   \tl_set:Nx \l__autotikz_hash_tl
     { \pdfmdfivesum { \detokenize \expandafter { \l__autotikz_code_tl } } }

   \tikzsetnextfilename { tikz-\l__autotikz_hash_tl }

   #1
 }
\ExplSyntaxOff

\numberwithin{equation}{section}

\def\ben{\begin{enumerate}}
\def\een{\end{enumerate}}

\def\bs{\bigskip}

\def\stubs{stubs of\xspace}

\numberwithin{theorem}{section}
\numberwithin{definition}{section}

\theoremstyle{theorem}

\makeatletter
\newtheorem*{rep@theorem}{\rep@title}
\newcommand{\newreptheorem}[2]{%
\newenvironment{rep#1}[1]{%
 \def\rep@title{#2 \ref{##1}}%
 \begin{rep@theorem}}%
 {\end{rep@theorem}}}
\makeatother

\newreptheorem{theorem}{Theorem}
\newreptheorem{lemma}{Lemma}

\newcommand{\rev}{\mathrm{Rev}}

\title{Quantum precomputation: parallelizing cascade circuits and the Moore--Nilsson conjecture is false}
\author[1]{Adam Bene Watts}
\affil[1]{University of Calgary}
\author[2]{Charles R. Chen}
\author[2]{J. William Helton}
\affil[2]{University of California, San Diego}
\author[3,4]{Joseph Slote}
\affil[3]{University of Washington}
\affil[4]{California Institute of Technology}
\date{}

\begin{document}
\clearpage\maketitle
\thispagestyle{empty}

\begin{abstract}
    Parallelization is a major challenge in quantum algorithms due to physical constraints like no-cloning.
    This is vividly illustrated by the conjecture of Moore and Nilsson from their seminal work on quantum circuit complexity~\cite[announced 1998]{MooreNilsson}: unitaries of a deceptively simple form---controlled-unitary ``staircases''---require circuits of minimum depth $\Omega(n)$.
    If true, this lower bound would represent a major break from classical parallelism and prove a quantum-native analogue of the famous $\mathsf{NC}\neq \mathsf{P}$ conjecture.
    
    In this work we settle the Moore--Nilsson conjecture in the negative by compressing all circuits in the class to depth $\mc O(\log n)$, which is the best possible.
    The parallelizations are exact, ancilla-free, and can be computed in poly($n$) time.
    We also consider circuits restricted to 2D connectivity, for which we derive compressions of optimal depth $\mc O(\sqrt{n})$.

    More generally, we make progress on the project of quantum parallelization by introducing a quantum blockwise precomputation technique somewhat analogous to the method of Arlazarov, Dini\v c, Kronrod, and Farad\v zev \cite{arlazarov1970economical} in classical dynamic programming, often called the ``Four-Russians method.''
    We apply this technique to more-general ``cascade'' circuits as well, obtaining for example polynomial depth reductions for staircases of controlled $\log(n)$-qubit unitaries.
\end{abstract}

\vspace{5em}
\begin{center}
    
\tikzfading[name=testfade,
            top color=transparent!0,
            bottom color=transparent!0,
            middle color=transparent!0
            ]
\begin{tikzpicture}
  \node (boxtext) [
  draw,rounded corners,
    fill=red!5,
    inner sep=8pt,
    text width=16cm,
    align=left,
  ] at (current page.north east)
  {\footnotesize {\color{purple}\textbf{Warning.}} Due to limitations of arXiv and the Ti\textit{k}Z \texttt{externalize} package, your PDF viewer may not correctly render transparencies in circuit diagrams below.
  A patched document is available at \url{joeslote.com/documents/precomputation.pdf}.};

  \node[
    fit=(boxtext),
    fill=white,
    path fading=testfade,
    fading angle=90
  ] {};
\end{tikzpicture}
\end{center}

\newpage

\setcounter{tocdepth}{2}

{\small
\tableofcontents
}

\newpage

\section{Introduction}

\noindent\textsc{Parallelism is a fundamental aspect of computation.}
On the one hand, finding ways to parallelize algorithms has obvious benefits for runtime in modern computer architectures.
On the other, it is desirable that \textit{some} programs do \textit{not} parallelize: various cryptographic primitives rely on the existence of so-called \textit{inherently sequential functions}---polytime functions which do not admit efficient polylog-depth circuits---see for example~\cite{BonehNaor00,CouteauRR21,bonneau2025}.\footnote{We do not attempt a survey of the sizeable literature on time-lock puzzles and related objects initiated by~\cite{RSW96}; the cited works provide examples where inherently sequential computations either are provably necessary or at least circumvent previous impossibility results.}

But despite its importance, parallelism remains poorly understood.
For example, it is open whether every classical bounded fan-in circuit composed of $m$ gates has an equivalent circuit of depth merely $\mathrm{polylog}(m)$ while keeping the gate count to $\mathrm{poly}(m)$.
This is essentially the $\mathsf{NC}$ vs. $\mathsf{P}$ question, a central problem in complexity theory that has stymied the community since its posing in 1981~\cite{cook}.
While it is generally expected that there are functions in $\mathsf{P}$ requiring $\poly(n)$-depth $\mathsf{NC}$ circuits, lower bounds have been stalled for over 30 years at $(3-o(1))\log n$, coming from H{\aa}stad's landmark formula lower bounds~\cite{hastad:shrinkage1993,doi:10.1137/S0097539794261556}.\footnote{Though there is an important program towards super-logarithmic lower bounds via the KRW conjecture~\cite{KRW95}; see the recent work \cite{meir25} for an overview.}

\subsubsection*{Quantum parallelism: constrained by physics?}
Parallelism in quantum computation is even more mysterious.
Not only are super-logarithmic lower bounds not known, but even basic approaches to upper bounds---borrowed from the classical literature on parallel algorithms---appear to be frustrated by constraints inherent to quantum physics.
To see this, consider the very simple toy example pictured in \Cref{fig:parallelization-challenge-Intro}.

\begin{figure}[ht]
    \centering
    \begin{subfigure}{.61\textwidth}
        \[
    \autoname{
    \begin{quantikz}[wire types={q,b,n,n}, classical gap=2pt]
        \lstick{$x$}& \gate{f}\wireoverride{b}
        &\coctrl{}& &\\
        \lstick{$y$}& &\gate{g}\wire[u]{q} &\setwiretype{q} &\\
        & & &\\
        &\slice[style = black, label style={pos=1, anchor=north}]{$T$} & \slice[style = black, label style={pos=1, anchor=north}]{$2T$} &
    \end{quantikz}}
    \: =
    \autoname{\begin{quantikz}[wire types={q,b,n,b}, classical gap=2pt]
        \lstick{$x$}& \wireoverride{b}&\gate{f}\wireoverride{b}&[1em]\coctrl{}\wire[d]{q}\setwiretype{q}&[1em]\\
        \lstick{$y$}& \gate[3, label style={rotate=90}][1.3em]{\mathclap{\text{COPY}}}&\gate{g_0} &\gate[3, label style={rotate=90}][1.3em][2em]{\mathclap{\text{SELECT}}}\setwiretype{q} & \setwiretype{n}\\[-.75em]
        &&&&\setwiretype{q}\\[-.75em]
        & \wireoverride{n} \slice[style = black, label style={pos=1, anchor=north}]{$1$} &\gate{g_1} \slice[style = black, label style={pos=1, anchor=north}]{$T + 1$} &\setwiretype{q} \slice[style = black, label style={pos=1, anchor=north}]{$T + 2$}& \setwiretype{n} 
    \end{quantikz}}
    \]
    \caption{Classical parallelization from depth $2T$ to depth $T + 2$.}
    \label{subfig:classical_parallel}
    \end{subfigure}
    \hfill\vline\hfill
    \begin{subfigure}{.38\textwidth}
        \[
        \begin{quantikz}[wire types={b,q,b,n,n}, classical gap=2pt]
         \lstick[2,brackets=none]{$\displaystyle\ket{\vphantom{\int}\psi \!}$\hspace{-0.7em}} & \gate[2]{F} & &\\[-1.3em]
        & &\coctrl{}&\\
        \lstick{\ket{\phi}} & &\gate{G}\wire[u]{q} & \\
        & & \\[20pt]
        & &
    \end{quantikz}\quad=\quad \text{\textbf{?}}
        \]
        \caption{Quantum parallelization?}
        \label{subfig:quantum_parallel}
    \end{subfigure}
    \caption{A standard idea in classical parallelization with no immediate quantum analogue.
    The half-open control denotes the product of an open control and closed control---see \Cref{sec:main results} for a formal explanation.
    }
    \label{fig:parallelization-challenge-Intro}
\end{figure}
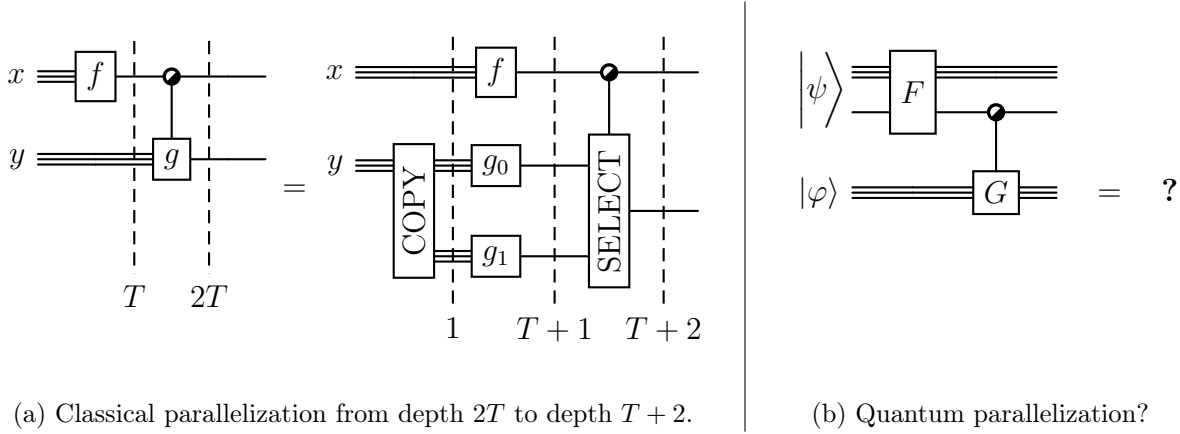

In the classical scenario, we wish to compute a function $h(x,y):=g_{f(x)}(y)$ on two classical inputs $x$ and $y$ for some Boolean functions $f,g_0,g_1$.
The naive implementation of $h$ is to first compute $b:=f(x)\in\{0,1\}$, and then run $g_b$ on $y$.
If $f,g_0,$ and $g_1$ each take time $T$ to run, this approach takes time about $2T$.
But $h$ can be easily parallelized: while one subcircuit is computing $f(x)$, we should copy input $y$ and evaluate $g_0(y)$ and $g_1(y)$ in parallel.
Then the final output can be simply selected (in constant time) based on the outcome of $f$.
Here we find parallelization essentially \emph{halves the runtime} to just $T$ plus a small constant overhead.
These techniques also extend naturally to reversible classical circuits.\footnote{By adding a layer of uncompute operations to our classical circuit, which can also be performed in parallel after the desired output bit is computed.}

Now consider a quantum variant: we replace strings $x$ and $y$ with unknown quantum states $\ket{\psi},\ket{\phi}$, and our goal is to apply unitary transformations $G_0$ or $G_1$ to $\ket{\phi}$, controlled by the last qubit of $F\ket{\psi}$.
Here the parallelization trick no longer works: because of the no-cloning theorem, we cannot copy $\ket{\phi}$ to simultaneously precompute $G_0\ket{\phi}$ and $G_1\ket{\phi}$.
We could try to ``guess'' the outcome of $F\ket{\psi}$, but if we guessed wrongly we would have to rewind our computation on $\ket{\phi}$ and begin anew.
Thus, it seems \emph{this classical parallelization technique has no immediate quantum analogue}.

\subsubsection*{The Moore--Nilsson conjecture}
By iterating the heuristic argument above, one may be led to wonder if there is no way to parallelize a sequence of controlled quantum operations from qubit $j-1$ to $j$, $j=2,\ldots, n$.
This is precisely the conjecture that concludes the seminal paper of Christopher Moore and Martin Nilsson \cite{MooreNilsson}, which 25 years ago defined the $\mathsf{QNC}$ hierarchy and initiated the study of concrete quantum circuit complexity.

\begin{conjecture*}[Moore--Nilsson \cite{MooreNilsson}]
Let $U^{(1)},U^{(2)},\cdots,U^{(n)}$ be 1-qubit unitaries which are neither diagonal nor anti-diagonal.
Then the following circuit has depth $\Omega(n)$.
\[ 
\mc C(\vec U)\quad=\qquad
\autoname{\begin{quantikz}[classical gap=1.8pt,row sep={2em,between origins},
wire types={q,q,q,q,q,q},transparency group, blend group=normal]
     & \ctrl{1} & & & & & \\
     & \gate{U^{(1)}} & \ctrl{1} & & & & \\
     & & \gate{U^{(2)}} & \ctrl{1} & & & \\
      \newwave[minimum height=3em]&&&\defer{\hspace{3em}\rotatebox{-45}{$\cdots$}}&&&\\[.5em]
    & & & & \gate{U^{(n-1)}}\wire[u]{q} & \ctrl{1} &  \\
    & & & & & \gate{U^{(n)}} & 
\end{quantikz}
}
.
\]
\end{conjecture*}
Observe that classical analogues of these circuits can be easily parallelized: for example imagine setting each $U^{(j)}$ to Pauli $X$.
This yields a classical reversible circuit implementing a prefix-sum computation, and such circuits are well-studied and have $\mc O(\log n)$ depth parallelizations~\cite{blelloch1990prefix,cole1989faster}.

The circuit class in the Moore--Nilsson conjecture, however, has resisted parallelization for the quarter century since its first announcement.
To the authors' best knowledge there is no published progress on the conjecture.
Through personal communications the authors are aware of some study made of the specific circuit $\mc C(H,H,\ldots, H)$.
Chinmay Nirkhe and colleagues derived a circuit that produces $\mc C(H,H,\ldots, H)\ket{0}$ in $\mc O(\log n)$ depth---though of course this says nothing about the circuit required to implement the entire unitary. 
Independently, Anne Broadbent considered the circuit $\mc C(H,H,\ldots, H)$ in the context of position verification.
This led Florian Speelman and colleagues to note that ``middle bits'' of the $\mc C(H,H,\ldots, H)$ unitary could be extracted up to small error by log-depth quantum circuits and, as a consequence, that it was possible to construct a log-depth circuit which approximately inverted the action of $\mc C(H,H,\ldots, H)$ on computational basis states.
However, it is not clear how to extend this approach to invert the action of $\mc C(H,H,\ldots, H)$ on general states or to obtain an operator-norm approximation of the unitary.

\bigskip

The Moore--Nilsson conjecture highlights dual motivations for studying parallelization in the quantum world:
\begin{itemize}
    \item The first is the need for \textit{upper bounds}: because Moore--Nilsson circuits are so simple, they underscore how underdeveloped quantum parallelization techniques are in comparison to the classical case.
    Can we at least recover quantum analogues of basic parallelization ideas? These could lead to new compilation techniques, reducing the runtime of both near-term and fault-tolerant quantum algorithms.
    \item The second is \textit{lower bounds}: the Moore--Nilsson conjecture presents a simple, concrete class of unitaries that might have super-logarithmic depth lower bounds.
    Could unitary circuit depth lower bounds join the growing list \cite[etc.]{kretschmer2021quantum,kretschmer2023quantum,DBLP:conf/stoc/LombardiMW24} of quantum hardness results that seem to be available independent of separations (or the lack thereof) in classical complexity?
\end{itemize}

In this work we develop new techniques for parallelizing quantum circuits. We use these techniques to strongly refute the Moore--Nilsson conjecture, showing Moore--Nilsson circuits can be exactly compiled with depths that match lightcone lower bounds.
We also use our parallelization techniques to obtain milder depth reductions for a larger class of circuits consisting of cascades of controlled multi-qubit unitaries.
Finally, we discuss new classes of simple unitaries which may still admit super-logarithmic depth lower bounds.

\subsection{Main results}
\label{sec:main results}
We will work with circuits consisting of arbitrary one and two qubit gates with all-to-all connectivity---\textit{i.e.,} $\mathsf{QNC}$ circuits\footnote{Equivalently (up to constant factors), we consider circuits with a gate set consisting of arbitrary one-qubit gates and at least one entangling two qubit gate~\cite{bullock2003arbitrary,bremner2002practical}.}---as well as circuits restricted to 2D connectivity, which we call $\mathsf{QNC}_\mathsf{2D}$ circuits.
All our results are constructive and can therefore be converted into approximate compilation schemes over a finite gate set using standard techniques.

The first contribution of this paper is a refutation of the Moore--Nilsson conjecture.

\begin{theorem}
\label{thm:single_qbit_parallelization-Intro}
    Every Moore--Nilsson unitary $\mc C(\vec U)$ on $n$ qubits is computed exactly by...
    \begin{itemize}
        \item A $\textsf{QNC}$ circuit of depth $\mc O(\log n)$ and no ancillae, and
        \item A $\textsf{QNC}_\textsf{2D}$ circuit of depth $\mc O(\sqrt n)$ and $\mc O(n)$ ancilla qubits.
    \end{itemize}
    Both of these depths are the best possible.
    Moreover, these circuits can be computed from the list of gates $U^{(1)},\ldots, U^{(n)}$ in time $\mathrm{poly}(n)$.
\end{theorem}

Our second result is a milder parallelization theorem that applies to a broader class of circuits which we now define.

\begin{definition}\label{defn:control_cascade}
Let 
$\vec{U} = \big(U_0^{(1)}, U_1^{(1)}, 
U_0^{(2)}, U_1^{(2)}, ..., U_0^{(m)}, U_1^{(m)} \big)$ 
be a list of $2m$ unitaries, each acting on $k$ qubits.
Then the \emph{control cascade} $\mc C (\vec{U})$ is the $(km+1)$-qubit unitary implemented by
\[\mc C(\vec U) \;=\;\mbox{
\tikzsetnextfilename{cascade-def}
\begin{quantikz}[wire types={q,b,q,b,q,n,b,q,b},classical gap=2pt, row sep={2em,between origins}]
     & \coctrl{} & & & & & \\
     & \gate[2]{U^{(1)}}\wire[u]{q} & & & & & \\[-.5em]
     & & \coctrl{} & & & & \\
     & & \gate[2]{U^{(2)}}\wire[u]{q} & & &  &  \\[-.5em]
     & & & \coctrl{}\wire[d]{q} & & &\\
    \newwave[minimum height=3em]&&&\defer{\hspace{3em}\rotatebox{-45}{$\cdots$}}&&&\\[.5em]
    & & & & \gate[2]{U^{(m-1)}}\wire[u]{q} &  &  \\[-.5em]
    & & & & & \coctrl{}\wire[d]{q} & \\
    & & & & & \gate{U^{(m)}}&
\end{quantikz}}\,.\]
\end{definition}

\medskip

\noindent\textit{Notation:} Here and throughout we use the ``closed-open'' control notation to refer to a product of open and closed controlled unitaries.
So, for example: 
\begin{align*}
\begin{quantikz}
    & \coctrl{1}\wire[d]{q} & \\
    & \gate{U} &
\end{quantikz} = 
\begin{quantikz}
    & \octrl{1} &  \ctrl{1} & \\
    & \gate{U_0} & \gate{U_1} & 
\end{quantikz}\,.
\end{align*}
We extend this notation in the natural way to multiply controlled gates, so
\begin{align*}
\begin{quantikz}
    & \coctrl{1}\wire[d]{q} & \\
    & \coctrl{1}\wire[d]{q} & \\
    & \gate{U} &
\end{quantikz} = 
\begin{quantikz}
    &\octrl{1} & \octrl{1} & \ctrl{1} &  \ctrl{1} &\\
    &\octrl{1} & \ctrl{1} &\octrl{1} & \ctrl{1} &\\ 
    &\gate{U_{00}} & \gate{U_{01}}& \gate{U_{10}}& 
    \gate{U_{11}}& 
\end{quantikz}.
\end{align*}
When necessary we will specify the sub-matrices of a closed-open controlled unitary by writing $U$ as a vector of unitaries, \textit{i.e.} $U = (U_0, U_1)$ or $U = (U_{00}, U_{01},U_{10},U_{11})$ respectively.
In text, we will use the phrase ``multiplexer $U$'' to refer to a unitary of this form.
\bigskip

One may note that, up to a layer of single qubit unitaries, Moore--Nilsson unitaries coincide with the set of control cascades with $k=1$.\footnote{Since we can always rewrite a multiplexer $U = (U_0,U_1)$ as a $U_0$ unitary followed by a controlled $U_0^{\dagger}U_1$.}
For the larger class of control cascade unitaries we are still able to obtain some depth reductions.

\begin{theorem}
\label{thm:general-exact-Intro}
        For any control cascade $\mc C(\vec U)$ with $m$-many $k$-qubit $U$'s there is a $\textsf{QNC}$ circuit with depth
        $\mc O(4^k + m2^k)$ and no ancillae which exactly computes $\mc C(\vec U)$. 
        Moreover, this circuit can be computed from $\mc C(\vec U)$ in time $\mathrm{poly}(m, 2^k)$.
\end{theorem}

To appreciate this theorem, consider the regime of $m=n$ and $k=\log_2(n)$ (so the unitary $\mc C(\vec U)$ acts on $n\log_2(n)$ qubits).
Because $k$-qubit unitaries require up to $\mc O(4^k)$ gates \cite{tucci1999rudimentary,mottonen04}, and $\mc C(\vec U)$ consists of $m$ of these in a cascade, a naive compilation of such unitaries would yield circuits of depth on the order of $m4^k = n^3$.
But in the same parameter regime we may apply \Cref{thm:general-exact-Intro} to get a polynomial improvement to depth~$\mc O (n^2)$. 

\begin{corollary}
\label{cor:misn-Intro}
    Consider $\mc C(\vec U)$ with $n$-many $\log_2(n)$-qubit gates.
    Then $\mc C(\vec U)$ has an exact, ancilla-free circuit with depth $\mc O(n^2)$.
    (\textit{C.f.} the naive $\mc O(n^3)$-depth circuit.)
\end{corollary}

We can decrease the depth further by allowing for ancillae and approximations.

\begin{theorem} 
\label{thm:general_parallelization-Intro}
    For any $\mc C(\vec{U})$ with $m$-many $k$-qubit unitaries and error threshold $\epsilon$ there exists:
    \begin{enumerate}[label = (\alph*)]
        \item A $\textsf{QNC}$ circuit of depth $\mc O(4^k + mk)$ and with $m2^k$ ancilla qubits which implements $\mc C(\vec{U})$ exactly. 
        \item A $\textsf{QNC}$ circuit of depth $\mc O(4^k + m\log \log(m/\epsilon))$ and with $m\log(m/\epsilon)$ ancilla qubits which implements a unitary $\mc C '$ satisfying $\|\mc C' - \mc C(\vec{U})\|_\infty \leq \epsilon$.
    \end{enumerate}
    Moreover, these circuits can be computed from $\mc C(\vec U)$ in time $\mathrm{poly}(m, 2^k)$.
\end{theorem}

\noindent\Cref{thm:general_parallelization-Intro} can give a near-quadratic depth improvement over naive techniques.
With $\mc C(\vec U)$ composed of $n$-many $\log_4(n)$-qubit gates, the naive depth is on the order of $n4^k=n^2$.
However direct computation in the same parameter regime gives:

\begin{corollary}
    Consider $\mc C(\vec U)$ with $n$-many $\log_4(n)$-qubit gates.
    Then there exists:
    \begin{enumerate}[label = (\alph*)]
    \item A $\textsf{QNC}$ circuit with depth $\mc O(n \log n )$ and $n^{3/2}$ ancillae implementing $\mc C(\vec{U})$ exactly.
    \item A $\textsf{QNC}$ circuit with depth $\mc O(n \log \log n)$ and $\mc O(n \log n )$ ancillae implementing a unitary~$\mc C'$ satisfying $\|\mc C' - \mc C(\vec{U})\|_\infty \leq 2^{- \Omega(n)}$.
    \end{enumerate}
    (\textit{C.f.} the naive $\mc O(n^2)$-depth circuit.)
\end{corollary}

We discuss situations in which our parallelization techniques do not immediately give depth reductions
in~\Cref{subsec:Outlook} of the introduction and then again in~\Cref{sec:discussion} of the main paper.

\subsection{Proof overviews}
\subsubsection{A quantum precomputation technique}

The starting point of our results is to essentially fill in the right-hand side of \Cref{subfig:quantum_parallel}, though in a weaker way than is achieved classically.
Classically, the identity in \Cref{subfig:classical_parallel} splits $g$ into a preprocessing operation independent of the control and then an $\mc O(1)$-time controlled operation.
In the quantum case, we do not obtain an $\mc O(1)$-time controlled operation, but we are able to reduce from a controlled general unitary to a controlled \textit{diagonal} unitary.
Concretely, we begin with the following lemma.



\begin{lemma}[Quantum precomputation identity, in brief]
\label{lem:q4rid-Intro}
Let $U_0, U_1$ be $k$-qubit unitaries. Then
\begin{align}
\label{def:openClosedR}
    \autoname{\begin{quantikz}[wire types = {q, b, q}, classical gap = 2pt, row sep = {2em, between origins}]
        & \octrl{2} & \ctrl{1} & \\
        & \gate[2]{U_0} & \gate[2]{U_1} & \\
        & \ghost{U_0} & \ghost{U_1} &
    \end{quantikz}}
    \quad = \quad 
    \autoname{\begin{quantikz}[wire types = {q, b, q}, classical gap = 2pt, row sep = {2em, between origins}]
        &      & \ctrl{1} &  & \coctrl{}\wire[d]{q} &\\
        & \gate[2]{P} & \gate[2]{D} & & \gate{R}   &\\
        & \ghost{P} & \ghost{D} & \gate{\Phi} & \coctrl{}\wire[u]{q}     &
    \end{quantikz}}
\end{align}
for some unitary $P$ and multiplexer $R$, diagonal unitary $D$, and a universal (\textit{i.e.}, fixed) unitary $\Phi$. 
\end{lemma}

The proof of this lemma is given in \Cref{subsec:Proof_of_Q4R_general}.
It employs a classical result in matrix analysis known as the cosine-sine decomposition (CS decomposition; see \textit{e.g.}, \cite{cssurvey}), which previously played an important role in quantum circuit compilation~\cite{tucci1999rudimentary,mottonen04} and more recently has provided a valuable perspective on the Quantum Singular Value Transform~\cite{DBLP:conf/sosa/TangT24}.

This precomputation identity is the starting point for all our results, which are best explained out of order.
\Cref{thm:general-exact-Intro} follows from applying \Cref{lem:q4rid-Intro} to each gate in the cascade, as we illustrate now for the case of $m=3$ controlled unitaries.

To parallelize  $\mc C(\vec U)$, begin by applying \Cref{lem:q4rid-Intro} to each controlled unitary $U^{(i)}$'s in the cascade:
    \begin{align}
    \label{eq:intro-q4r-sketch}
    \adjustbox{scale=0.8}{\autoname{\begin{quantikz}[wire types={q,b,q,b,q,b,q},classical gap=2pt, row sep={2em,between origins}]
     & \coctrl{} &[-0.8em] &[-0.8em] & \\
     & \gate[2]{U^{(1)}}\wire[u]{q} & & & \\[-.5em]
     & & \coctrl[coctrl color=blue!70!black]{} & & \\
     & & \gate[2,style={fill=blue!10, draw=blue!70!black},label style=blue!50!black]{U^{(2)}}\wire[u,style={draw=blue!70!black}]{q} & &   \\[-.5em]
     & & & \coctrl{}\wire[d]{q} &\\
    & & & \gate[2]{U^{(3)}} &  \\[-.5em]
    & & & & 
\end{quantikz}}}
    \quad&=\quad \adjustbox{scale=0.8}{\autoname{\begin{quantikz}[wire types={q,b,q,b,q,b,q},classical gap=2pt, row sep={2em,between origins}]
        & &[-0.5em]\ctrl{1} &[-0.5em] &[-1em] \coctrl{}\wire[d]{q}&[-0.5em] &[-0.5em] &[-1em] &[-.5em] &[-.5em] &[-1em] &\\
        & \gate[2]{P^{(1)}} &\gate[2]{D^{(1)}} & & \gate{R^{(1)}}& & & & & & &\\
        & & &\gate{\Phi} &\coctrl{}\wire[u]{q}&\ctrl[style={fill=blue!70!black, draw=blue!70!black}]{1} & & \coctrl[coctrl color=blue!70!black]{}\wire[d,style={draw=blue!70!black}]{q}& & & &\\
        &  & &&\gate[2,style={fill=blue!10, draw=blue!70!black},label style=blue!50!black]{P^{(2)}} &\gate[2,style={fill=blue!10, draw=blue!70!black},label style=blue!50!black]{D^{(2)}}  & &\gate[style={fill=blue!10, draw=blue!70!black},label style=blue!50!black]{R^{(2)}}& & & &\\
        & && & &  &\gate[style={fill=blue!10, draw=blue!70!black},label style=blue!50!black]{\Phi}&\coctrl[coctrl color=blue!70!black]{}\wire[u,style={draw=blue!70!black}]{q}&\ctrl{1} & &\coctrl{}\wire[d]{q}& \\
        & & && & & &\gate[2]{P^{(3)}} &\gate[2]{D^{(3)}}&&\gate{R^{(3)}}& \\
        & & & & & && & & \gate{\Phi} &\coctrl{}\wire[u]{q} &
    \end{quantikz}}},
    \end{align}
where we have highlighted the second unitary and its replacement for visual clarity.
The resulting $P^{(i)}$'s act on disjoint qubits so they can be computed in parallel, and the $R^{(i)}$'s only overlap on control qubits and therefore commute.
Rearranging based on these observations, we find the circuit splits into three ``stages'':
\begin{align*}
    \eqref{eq:intro-q4r-sketch} \quad&=\quad\adjustbox{scale=0.8}{\autoname{\begin{quantikz}[wire types={q,b,q,b,q,b,q,b,q},classical gap=2pt, row sep={2em,between origins}]
        &\slice[style = black]{} &[1.5em] \ctrl{1} &[-.5em] &[-1em] &[-.5em] &[-1em] &[-.5em] \slice[style = black]{}&[1.5em] \coctrl{}\wire[d]{q}&[-1.5em] &\\
        & \gate[2]{P^{(1)}} &\gate[2]{D^{(1)}} & & & & && \gate{R^{(1)}}&&\\
        & &  &\gate{\Phi}&\ctrl[style={fill=blue!70!black, draw=blue!70!black}]{1} & & & &\coctrl{}\wire[u]{q}&\coctrl[coctrl color=blue!70!black]{}\wire[d,style={draw=blue!70!black}]{q}&\\
        & \gate[2,style={fill=blue!10, draw=blue!70!black},label style=blue!50!black]{P^{(2)}} & & &\gate[2,style={fill=blue!10, draw=blue!70!black},label style=blue!50!black]{D^{(2)}} & & & & &\gate[style={fill=blue!10, draw=blue!70!black},label style=blue!50!black]{R^{(2)}}&\\
        & & & & & \gate[style={fill=blue!10, draw=blue!70!black},label style=blue!50!black]{\Phi} & \ctrl{1}& &\coctrl{}\wire[d]{q}&\coctrl[coctrl color=blue!70!black]{}\wire[u,style={draw=blue!70!black}]{q}&\\[.2em]
        &\gate[2]{P^{(3)}}& & & & &\gate[2]{D^{(3)}} & &\gate{R^{(3)}} & &\\
        & & & & & & &\gate{\Phi} &\coctrl{}\wire[u]{q} &&
    \end{quantikz}}}
    \end{align*}

The $P^{(i)}$'s are $k$-qubit gates acting in parallel,
so by standard circuit compilation results
\cite{tucci1999rudimentary,vartiainen04} the first stage has a $\mathsf{QNC}$ circuit of depth 
$=\mc O(4^k)$.
The controlled $R^{(i)}$'s can always be organized into two layers, so they also parallelize to depth $\mc O(4^k)$.
Diagonal unitaries on $k$ qubits only ever require $\mc O(2^k)$ gates \cite{bullockOptDiag}, so the middle section requires depth at most $\mc O(m 2^k)$.
Summing these estimates yields 
\Cref{thm:general-exact-Intro}. A formal proof is given in~\Cref{subsec:general_control_cascade_reduction}.

\medskip

\Cref{thm:general_parallelization-Intro} builds on this approach by observing that diagonal unitaries are themselves amenable to classical precomputation techniques.
Using ancilla qubits to precompute the ``truth tables'' of each of the $D^{(j)}$'s in the cascade, either exactly or to finite precision, gives \Cref{thm:general_parallelization-Intro}.
The proof of this theorem appears in \Cref{subsec:diagonal_caching}.

\medskip

We remark that theorems \Cref{thm:general-exact-Intro,thm:general_parallelization-Intro} are already enough to give a weak disproof of the Moore--Nilsson conjecture.
The first step is to observe that a controlled cascade of $n$ single-qubit unitaries can also be viewed as a controlled cascade consisting of fewer blocks of multi-qubit unitaries, for example as follows.
\[
\autoname{\begin{quantikz}[wire types={q,q,q,q,q,q,q,},classical gap=2pt, row sep={2em,between origins}]
    & \ctrl{1} &[-1em] &[-1em] &[-1em] &[-1em] &[-1em] &\\
    & \gate{U^{(1)}} \gategroup[2,steps=2,style={inner
sep=2pt, fill=blue!10,dashed, rounded corners},background]{}             &    \ctrl{1}                                                                               &     
    &                               &                       &                       &                                       \\[-.2em]
    &                   &        \gate{U^{(2)}}                                          &   \ctrl{1}             &  
                                  &                       &                       &                                       \\[1.3em]
    &                               &                                              &                                         
    \gate{U^{(3)}} \gategroup[2,steps=2,style={inner
sep=2pt, fill=blue!10,dashed, rounded corners},background]{}         & \ctrl{1}                      &                       &                                                              
    &                               \\[-.2em]
    &                                                                     &                                       & 
    &                                                   \gate{U^{(4)}}    &   \ctrl{1}                                                            
    &           &                 \\[1.3em]
    &                               &                                              &                                       &                              
    &                                                                                                                                 
                \gate{U^{(5)}} \gategroup[2,steps=2,style={inner
sep=2pt, fill=blue!10,dashed, rounded corners},background]{}&   \ctrl{1}               &               
\\
    &                                                                                                           &                              
    &                               &                       &                       &                                                                         \gate{U^{(6)}}                                            & 
\end{quantikz}} \quad=:\quad \autoname{\begin{quantikz}[wire types={q,q,q,q,q,q,q},classical gap=2pt, row sep={2em,between origins}]
     & \coctrl{} &[-0.8em] &[-0.8em] & \\
     & \gate[2]{V^{(1)}}\wire[u]{q} & & & \\[-.5em]
     & & \coctrl{} & & \\
     & & \gate[2]{V^{(2)}}\wire[u]{q} & &   \\[-.5em]
     & & & \coctrl{}\wire[d]{q} &\\
    & & & \gate[2]{V^{(3)}} &  \\[-.5em]
    & & & & 
\end{quantikz}}
\]%
Applying \Cref{thm:general_parallelization-Intro} to $\mc C(\vec V)$ yields a $1/\poly(n)$ operator-norm approximation to the original Moore--Nilsson circuit with depth $\mc O(n \log\log n /\log n )$ and $\tilde{\mc O}(n)$ ancillae.
This is explained in more detail in \Cref{subsec:chunking}.

\subsubsection{Aside: comparing with classical work}

Before explaining the proof of \Cref{thm:single_qbit_parallelization-Intro}, we pause to mention that the (blockwise) precomputation methods described thus far can be viewed as a quantum analogue of a well-known classical technique for speeding up dynamic programming introduced by the Soviet researchers Arlazarov, Dini\v c, Kronrod, and Farad\v zev \cite{arlazarov1970economical}.
Arlazarov et al. showed how to speed up exploration of a $d$-dimensional memo table by precomputing the ``truth tables'' of small blocks, leading to time complexity improvements from the naive bound of $\mc O(n^d)$ to the asymptotically better $\mc O(n^d/\log^{C(d)}(n))$.
This approach has come to be known in the algorithms literature as the ``Four-Russians Method'' or the ``Four-Russians speedup'' even though only one---Arlazarov---was actually Russian \cite{Gusfield1997}.

That said, the analogy of our quantum results to \cite{arlazarov1970economical} is not yet a complete one.
Arlazarov et al. were not considering parallel algorithms (their speedup holds also for standard Turing machines), and when one switches to circuits a few remarks should be made about the difference between depth and gate complexity for computing memo tables.
Another point is that the circuits we parallelize---control cascades---are essentially a quantum analogue of \textit{one-dimensional} memo tables.
The extent to which the analogy to \cite{arlazarov1970economical} continues into the full setting of  many-dimensional memo tables is a tantalizing open question.
We expand on these points in the discussion, \Cref{sec:discussion}.

\subsubsection{Optimal-depth circuits for Moore--Nilsson unitaries}

It turns out that we can do much better for Moore--Nilsson circuits than described above.
The key is to exploit special structure in the precomputation identity that appears when it is applied to 1-qubit control cascades (the $V^{(j)}$'s above).

\begin{lemma}[Quantum precomputation identity for Moore--Nilsson circuits]
\label{lem:q4rid-moore-nilsson-Intro}
Let $U^{(1)},\ldots, U^{(\ell)}$ be any $1$-qubit unitaries.
Then there exists an $\ell$-qubit unitary $P$, a 1-qubit unitary $Q$, and a multiplexer $R$
such that
\begin{equation}
\label{eq:q4rmn}
    \autoname{\begin{quantikz}[wire types={q,q,q,q},classical gap=2pt, row sep = {1.75em, between origins}]
        & \ctrl{1} & & & &\\
        & \gate{U^{(1)}} & \ctrl{1} & & & \\
        & & \gate{U^{(2)}} & \ctrl{1} & &\\
        \newwave &&&\defer{\hspace{2em}\rotatebox{-45}{$\cdots$}}& &\\
        & & & & \gate{U^{(\ell)}}\wire[u]{q} &
    \end{quantikz}}
    \quad=\quad
    \autoname{\begin{quantikz}[wire types={q,b,q},classical gap=2pt]
        & & \ctrl{2} &\coctrl{}\wire[d]{q} &\\
        & \gate[2]{P} & & \gate{R} &\\
        & & \gate{Q} & \coctrl{}\wire[u]{q} &
    \end{quantikz}}
\end{equation}
\end{lemma}
\noindent This lemma resembles \Cref{lem:q4rid-Intro} except that the many-qubit controlled-diagonal gate has been replaced by a single-qubit controlled gate.
The proof of \Cref{lem:q4rid-moore-nilsson-Intro} is given in~\Cref{subsec:proof_of_MNQ4R}, the main technical ingredient of which is a novel cosine-sine-type decomposition specialized to a class of circuits we term ``valley circuits.''

To see how \Cref{lem:q4rid-moore-nilsson-Intro} leads to \Cref{thm:single_qbit_parallelization-Intro}, let us repeat the procedure above with the stronger precomputation identity.
Beginning with a Moore--Nilsson circuit, apply \Cref{lem:q4rid-moore-nilsson-Intro} to each of the $V^{(j)}$'s and then commute the pre- and postprocessing unitaries to the left and right of the circuit as in \eqref{eq:intro-q4r-sketch}.
We obtain a circuit of the following form:
\[\mc C(\vec U)=\autoname{\begin{quantikz}[wire types={q,b,q,b,q,n,b,q,b,q},classical gap=2pt, row sep={2em,between origins}, column sep = {2em, between origins}]
    & &[2em] \ctrl{2} \gategroup[8,steps=4,style={inner
sep=2pt, fill=blue!10,dashed, rounded corners},background,label style={label
position=south,anchor=north,yshift=-0.5em}]{$=\mc C(\vec Q)$}& & & & &\coctrl{}\wire[d]{q}&& \\
    & \gate[2]{P^{(1)}} & & & & & &\gate{R^{(1)}}&&\\
    & & \gate{Q^{(1)}} &\ctrl{2} & & & &\coctrl{}\wire[u]{q}&\coctrl{}\wire[d]{q}&\\
    & \gate[2]{P^{(2)}} &  & & & & &&\gate{R^{(2)}}&\\
    & & & \gate{Q^{(2)}}& \ctrl{1}& & &\coctrl{}\wire[d]{q}&\coctrl{}\wire[u]{q}&\\
    & \defer{\rotatebox{90}{$\cdots$}}& \newwavepartial[style={minimum height=2.5em, wave color=blue!10},steps=4]{} && \defer{\hspace{2em}\rotatebox{-45}{$\cdots$}}& \ctrl{2}& \newwavepartial[style={minimum height=2.5em, wave color=white},steps=4]{}&\defer{\hspace{2em}\rotatebox{90}{$\cdots$}}&\\[.2em]
    & \gate[2]{P^{(m)}}& & & & & & &\gate{R^{(m)}}\wire[u]{q} & \\
    & & & & & \gate{Q^{(m)}}& & &\coctrl{}\wire[u]{q} &
\end{quantikz}}\,.
\]
We see a Moore--Nilsson circuit has appeared in the center column, this time on a small fraction of the original $n$ qubits.
Choosing a block size of $\mc O(1)$ and iterating this procedure $\mc O(\log n)$ times yields a $\mathsf{QNC}$ circuit of total depth $\mc O(\log n)$, as claimed in~\Cref{thm:single_qbit_parallelization-Intro}.
To obtain a $\mathsf{QNC_{2D}}$ circuit, one carefully arranges the $\log$-depth $\mathsf{QNC}$ circuit into a 2D grid, borrowing some tree embedding ideas from VLSI design~\cite{patersonRuzzoSnyder, Ruzzo1981}.
The formal proof is given in~\Cref{subsec:optimal_mn_circuits}.

\subsection{Outlook}
\label{subsec:Outlook}
The present work offers techniques to reduce the depth of quantum control cascade circuits in various parameter regimes.
This is good news for circuit compilation and quantum algorithms.
On the other hand, these parallelization results threaten the hope---first articulated by Moore and Nilsson---that it might be easier to resolve the unitary analogue of $\cc{NC}$ vs. $\cc{P}$ because even basic quantum circuits seemingly cannot be parallelized.

Our results here indicate that much remains to be understood.
While the main results tell us parallelization is sometimes easier than previously thought, our approach does not immediately capture a variety of generalizations of Moore--Nilsson circuits, such as:
\begin{itemize}
\itemsep0em
\item \textit{Cascade circuits with more controls.} Our \textit{control cascade} circuits have one qubit of control.
As soon as more qubits---or even one \textit{qutrit}---is used to control the next unitary, our techniques do not seem to immediately apply. For example, in the case of a cascade of qutrit-controlled unitaries, repeating the ideas we used above seems to require a certain $3\times 3$ generalization of the CS decomposition---and as we show in the discussion, \Cref{sec:discussion}, this generalization is false.
\item \textit{Multidimensional ``quantum memo tables.''} Control cascade circuits are in some sense a quantum analogue of a one-dimensional memo table.
What about two- or higher-dimensional ``quantum dynamic programming''? How far does the analogy to Arlazarov et al. \cite{arlazarov1970economical} go?
\item \textit{Non-control cascades.} 
    Our techniques also do not seem to immediately apply when control-$U$ is replaced by a general $(k+1)$-qubit unitary rather than $k$-qubit \textit{controlled} unitaries.
\end{itemize}
\noindent These items are discussed in more detail in \Cref{sec:discussion}.
The authors are hopeful that studying these generalizations will lead either to further developments in quantum parallelization methods or, potentially, to super-logarithmic unitary depth lower bounds.

\subsection{Comments on numerical stability and uniformity}
Our algorithms are exact and polynomial time in the model of real valued computation (BSS) with access to an SVD oracle (or more accurately, a CS decomposition oracle).
There are efficient and backwards-stable algorithms for computing the CS decomposition \cite{backwardcs}.
A full discussion of numerical particulars is beyond the scope of this work, but from this we may conclude the existence of an ``essentially exact'' polytime algorithm for standard Turing machines working to finite precision via the above paper on CS decompositions.
Here ``essentially exact'' means the runtime is polynomial in both the number of qubits $n$ and $\mathrm{polylog}(1/\eta)$ where $\eta$ is the desired output precision in operator norm.

\subsection{Other related work}

There is a preexisting work on ``Quantum Dynamic Programming'' \cite{QDP}, though their notion is different. In that work the next unitary in a sequence is not controlled by the previous state in the sense of being block diagonal in tensor product space, but instead is an arbitrary function of the previous state.
In this setting there is a great deal of difficulty in getting efficient circuits for even the naive staircase-type implementation. 
The authors of~\cite{QDP} show circuits of this form can be implemented approximately in linear depth, provided they have access to an exponential number of ancillae.

Another recent result~\cite{kahanamoku2025log} shows parallelizations in which ``staircase-like'' circuits make an appearance, this time for the Quantum Fourier Transform (QFT).
The authors of~\cite{kahanamoku2025log} take advantage of the ``staircase-like'' form of the approximate QFT to approximately commute subcircuits past each other, deriving an ancilla-free, logarithmic-depth circuit which implements the QFT up to small error in \textit{Frobenius} norm.
Then they show a logarithmic-depth worst-to-average case reduction circuit for the QFT, which together with the previous yields a logarithmic-depth operator norm approximation to the QFT with $\mc O(n/\log n)$ ancillae.
It is an interesting open question whether the techniques  in~\cite{kahanamoku2025log} can be adapted to the setting of the current paper or vice versa.
Naively, however, the settings appear incompatible: we derive exact or approximate parallelizations in the operator norm directly for our class of circuits, and it is unclear whether shallow worst-to-average case reductions are available for control cascade circuits.
In the other direction, our results only apply to staircase-like circuits with one qubit of control, while the staircase encountered in~\cite{kahanamoku2025log} has multi-qubit controlled phase rotations from one step to the next.

\subsection*{Acknowledgments}
Adam Bene Watts and Joseph Slote are very grateful to Atul Singh Arora for introducing the Moore--Nilsson conjecture to us.
Arora arrived at the Moore--Nilsson conjecture independently during his work on depth hierarchy theorems relative to an oracle~\cite{arora23,arora22}.
Adam Bene Watts and Joseph Slote are also very grateful for several discussions about this problem with Henry Yuen and the Columbia quantum group.
Adam Bene Watts would additionally like to thank
Umesh Vazirani, Richard Cleve, Chinmay Nirkhe, and Peter H\o yer for helpful comments about precursors to this work.
Parts of this project were completed while Joseph Slote was supported by Chris Umans' Simons Investigator grant.

\section{Quantum precomputation}
\label{sec:Q4R}

This section proves results about parallelizing general control cascade circuits. \Cref{subsec:Proof_of_Q4R_general} proves~\Cref{lem:q4rid-Intro}, a quantum precomputation identity that will be a key ingredient in subsequent proofs.
\Cref{subsec:general_control_cascade_reduction} proves~\Cref{thm:general-exact-Intro}, which shows we can rewrite a general control cascade as a control cascade of diagonal unitaries sandwiched between pre- and postprocessing operations. \Cref{subsec:diagonal_caching} discusses techniques for reducing the depth of the control cascade of diagonal unitaries by introducing ancillae, proving~\Cref{thm:general_parallelization-Intro}.

\subsection{The quantum precomputation identity}
\label{subsec:Proof_of_Q4R_general}

We begin by stating~\Cref{lem:q4rid-Intro}, with more detail than was included in the introduction.

\begin{replemma}{lem:q4rid-Intro}[Repeated]
Let $U=(U_0,U_1)$ be a $k$-qubit controlled unitary.
Then there exists $(k-1)$-qubit unitary $P$, $(k-1)$-qubit \textit{diagonal} unitary $D$,  and multiplexer $R$ unitary on $k+1$ qubits with two controls such that
\[
    \autoname{\begin{quantikz}[wire types = {q, b, q}, classical gap = 2pt, row sep = {2em, between origins}]
        & \coctrl{}\wire[d]{q}  & \\
        & \gate[2]{U} & \\
        & &
    \end{quantikz}}
    \;=\;
    \autoname{\begin{quantikz}[wire types = {q, b, q}, classical gap = 2pt, row sep = {2em, between origins}]
        &      & \ctrl{1} &  & \coctrl{}\wire[d]{q} &\\
        & \gate[2]{P} & \gate[2]{D} & & \gate{R}   &\\
        & \ghost{P} & \ghost{D} & \gate{\Phi} & \coctrl{}\wire[u]{q}     &
    \end{quantikz}}\,.
\]
Here $\Phi:=\frac1{\sqrt2}\left(\begin{smallmatrix}
    1 & i \\ i & 1
\end{smallmatrix}\right)=\sqrt{Z}H\sqrt{Z}$.

Moreover, the diagonal matrix $D$ has the following explicit form.
Let $W=\left(\begin{smallmatrix}
    W_{00} & W_{01}\\
    W_{10} & W_{11}
\end{smallmatrix}\right)$ be the $2\times 2$ block matrix corresponding to $\rev(U_0^\dagger U_1)$, \textit{i.e.}, $U_0^\dagger U_1$ with qubit order reversed.
Then \[D = \rev\begin{pmatrix}
       \Sigma_1 - i \Sigma_2 & 0\\
       0 & \Sigma_1 + i \Sigma_2
   \end{pmatrix}\]
where $\Sigma_1$ is a diagonal matrix of the singular values of $W_{00}$ in ascending order and $\Sigma_2$ is a diagonal matrix of the singular values of $W_{01}$ in descending order.
\end{replemma}

The proof of this lemma relies on a classical SVD-type decomposition for $2\times 2$ block matrices called the \textit{cosine-sine decomposition} (CS decomposition).
We state a version of this decomposition for unitaries.

\begin{lemma}[Cosine-sine decomposition]
\label{lem:CSD}
Consider any $2N \times 2N$ unitary
    \[U = \begin{pmatrix}
        U_{00} & U_{01}\\
        U_{10} & U_{11}
    \end{pmatrix}\,.
    \]
    Then there exist $N\times N$ unitaries $S_0,S_1,T_0,T_1$ such that
    $$
    U
    =
     \begin{pmatrix}
        S_0 & 0\\
        0 & S_1
    \end{pmatrix}
    \begin{pmatrix}
        \Sigma_1 & \Sigma_2\\
        - \Sigma_2 &   \Sigma_1
    \end{pmatrix}\,
    \begin{pmatrix}
        T_0 & 0\\
        0 & T_1
    \end{pmatrix}
    $$
    where $\Sigma_1$ is a diagonal matrix of the singular values of $U_{00}$ in ascending order, and $\Sigma_2$ is a diagonal matrix of the singular values of $U_{01}$ in descending order.
    Moreover, given $\Sigma_2$, this equivalence can be found exactly in time $\poly(N)$.
\end{lemma}

The CS decomposition has a long history.
The name and full statement is essentially due to Chandler Davis and William Kahan \cite{Davis1969SomeNB}, though ideas along these lines can be traced back as far as the work of Jordan from 1875 \cite{BSMF_1875__3__103_2}.
We refer the reader to \cite{cssurvey} for a historical overview.

The CS decomposition admits a quantum circuit interpretation.

\begin{corollary}[Cosine-sine circuit decomposition]
\label{cor:circuit_CS_decomp}
Any $k$-qubit unitary $U$ can be written as
\[\mbox{\begin{quantikz}[wire types={q,b},classical gap=2pt, row sep={2em,between origins}]
&  \gate[2]{U} & \\
&  \ghost{U} & 
\end{quantikz}} 
= \mbox{\begin{quantikz}
[wire types 
= {q, b}, classical gap = 2pt, row sep={2em,between origins}]
& \coctrl{}\wire[d]{q} & \gate{\Phi^\dagger} & \gate[2]{D} & \gate{\Phi} & 
 \coctrl{}\wire[d]{q} & \\
 & \gate{S} & & \ghost{D} & & \gate{T} &  
\end{quantikz}}
\]
with $\Phi 
   = \frac 1 {\sqrt{2}}
   \begin{pmatrix}
        1 & i \\
        i & 1
    \end{pmatrix},$
 $D$ a diagonal matrix, and multiplexers $S$ and $T$.
\end{corollary}

\begin{proof}
From \Cref{lem:CSD}, we set $\begin{quantikz}
    \coctrl{}&\gate{S}
\end{quantikz} = \begin{pmatrix}
        S_0 & 0\\
        0 & S_1
    \end{pmatrix}$, $\begin{quantikz}
    \coctrl{}&\gate{T}
\end{quantikz} = \begin{pmatrix}
        T_0 & 0\\
        0 & T_1
    \end{pmatrix}$, and $D = \begin{pmatrix}
       \Sigma_1 - i \Sigma_2 & 0\\
       0 & \Sigma_1 + i \Sigma_2
   \end{pmatrix}$. It is then straightforward to check that 
   \begin{align*}
       \left(\Phi^\dagger \otimes I_{2^{k-1}}\right) D \left(\Phi\otimes I_{2^{k-1}}\right) 
       = \begin{pmatrix}
        \Sigma_1 & \Sigma_2\\
        - \Sigma_2 &   \Sigma_1
    \end{pmatrix},
   \end{align*} 
   and we are done. 
\end{proof}

We now move on to the proof of \Cref{lem:q4rid-Intro}.

\begin{proof}[Proof (\Cref{lem:q4rid-Intro}):]
We begin by rewriting $\begin{quantikz}
    \coctrl{}&\gate{U}
\end{quantikz}$
as a circuit with a single control gate by ``guessing'' that $U_0$ is applied and then undoing if in fact the control wire is $1$:
\begin{align*}
\begin{quantikz}[wire types={q,b},classical gap=2pt, row sep={2em,between origins}]
&\octrl{1} & \ctrl{1} & \\
& \gate[2]{U_0}  & \gate[2]{U_{1}} & 
\end{quantikz}
&= \begin{quantikz}[wire types={q,n,b},classical gap=2pt, row sep={2em,between origins}]
& &&[-1em] \ctrl{1}&[-1em]  & \\[-1.2em]
& & & & & \\[-6pt]
&\gate{U_0} & \gate[style={line width=.5pt}]{U_0^\dagger}\gategroup[1,steps=3,style={inner
sep=2pt, fill=black!2},background]{}  & & \gate[style={line width=.5pt}]{U_{1}} & 
\end{quantikz}
\end{align*}
Applying the CS decomposition upside-down to $U_0^\dagger U_1$ gives
\begin{align}
\autoname{\begin{quantikz}[wire types={q,b},classical gap=2pt, row sep={2em,between origins}]
&\octrl{1} & \ctrl{1} & \\
& \gate[2]{U_0}  & \gate[2]{U_{1}} & 
\end{quantikz}}
&= \autoname{\begin{quantikz}[wire types = {q,n,b, q}, classical gap = 2pt]
& & & &\ctrl{1} & & & \\[-.5em]
& & & & & & & \\[-8pt]
& \gate[2]{U_0} & \gate[style={line width=.5pt}]{S}\gategroup[2,steps=5,style={inner
sep=2pt, fill=black!2},background]{} & & \gate[2,style={line width=.5pt}]{D} & & \gate[style={line width=.5pt}]{T} & \\
& & \coctrl{}\wire[u]{q}&\gate[style={line width=.5pt}]{\Phi^\dagger} &  & \gate[style={line width=.5pt}]{\Phi} & 
 \coctrl{}\wire[u]{q} & 
\end{quantikz}}\label{eq:cs-applied}
\end{align}
for some $S,T$ and diagonal $D$.
Next the outer control is distributed to the gates inside its target and the $\begin{quantikz}
    \phase{}&\gate{S}&\coctrl{}
\end{quantikz}$ gate is rewritten so that it is open-controlled.
It then commutes with the gates to its right, making way for the final rearrangement.
\begin{align*}
\eqref{eq:cs-applied}&=\begin{quantikz}[wire types = {q,b, q}, classical gap = 2pt]
& &\ctrl{1} & &\ctrl{1} & & \ctrl{1}& \\
& \gate[2]{U_0} &\gate{S} & & \gate[2]{D} & & \gate{T} & \\
& & \coctrl{}\wire[u]{q}& \gate{\Phi^\dagger} &  & \gate{\Phi} & 
 \coctrl{}\wire[u]{q} & 
\end{quantikz}\\
&=\begin{quantikz}[wire types = {q,b, q}, classical gap = 2pt]
& & &\octrl{1} & &\ctrl{1} & & \ctrl{1}& \\
& \gate[2]{U_0} & \gate{S} &\gate{S^\dagger} & & \gate[2]{D} & & \gate{T} & \\
& & \coctrl{}\wire[u]{q}&\coctrl{}\wire[u]{q}& \gate{\Phi^\dagger} &  & \gate{\Phi} & 
 \coctrl{}\wire[u]{q} & 
\end{quantikz}\\
&=\begin{quantikz}[wire types = {q,b, q}, classical gap = 2pt]
& &  & &\ctrl{1} & &\octrl{1}\gategroup[3,steps=2,style={inner
sep=2pt, fill=blue!10,draw=blue!70!black, dashed, rounded corners},background,label style={label
position=below,anchor=north,yshift=-0.2cm}]{} & \ctrl{1}& \\
& \gate[2]{U_0}\gategroup[2,steps=3,style={inner
sep=2pt, fill=yellow!10, draw=yellow!70!black, dashed, rounded corners},background,label style={label
position=below,anchor=north,yshift=-0.2cm}]{} & \gate{S}  & & \gate[2]{D} & &\gate{S^\dagger}& \gate{T} & \\
& & \coctrl{}\wire[u]{q}& \gate{\Phi^\dagger} &  & \gate{\Phi} & 
 \coctrl{}\wire[u]{q}&\coctrl{}\wire[u]{q} & 
\end{quantikz}\;=:\;\autoname{\begin{quantikz}[wire types = {q,b, q}, classical gap = 2pt,row sep={1cm,between origins}]
&  &\ctrl{1} & &\coctrl[coctrl color=blue!70!black]{}\wire[d,style={draw=blue!70!black}]{q} & \\
& \gate[2,style={fill=yellow!10, draw=yellow!70!black},label style=yellow!50!black]{P} &  \gate[2]{D} & &\gate[style={fill=blue!10, draw=blue!70!black},label style=blue!50!black]{R}& \\
& & & \gate{\Phi} & 
 \coctrl[coctrl color=blue!70!black]{}\wire[u,style={draw=blue!70!black}]{q}& 
\end{quantikz}}
\end{align*}

\end{proof}

\subsection{Reducing control cascades to nearly diagonal cascades}
\label{subsec:general_control_cascade_reduction}

Now we move on to the proof of~\Cref{thm:general-exact-Intro}, the key ideas of which were sketched in the introduction.

\begin{reptheorem}{thm:general-exact-Intro}[Repeated]
        For any control cascade $\mc C(\vec U)$ with $m$-many $k$-qubit $U$'s there is a $\mathsf{QNC}$ circuit with depth
        $\mc O(4^k + m2^k)$ and no ancillae which exactly computes $\mc C(\vec U)$.
\end{reptheorem}

\begin{proof}

Recall that the control cascade $\mc C(\vec{U})$ consists of a sequence of $m$ closed-open controlled unitaries, each acting on $k+1$ bits, with the first bit of each closed-open controlled unitary (the control bit) overlapping the last bit of the unitary previous. See~\Cref{defn:control_cascade} for a picture. Applying \Cref{lem:q4rid-Intro} to each controlled unitary in this cascade and then collecting like terms produces three layers of unitaries: a layer of $k$ qubit ``preprocessing'' unitaries $P^{(i)}$, a cascade of $k$-qubit diagonal unitaries and single qubit rotations, and then a layer of doubly controlled $k + 1$ qubit ``postprocessing'' unitaries $R^{(i)}$:

\begin{align*}
\mc C(\vec U)&=\quad \autoname{\begin{quantikz}[wire types={q,b,q,b,q,n,b,q,b,q},classical gap=2pt, row sep={2em,between origins}]
    & &[-0.5em]\ctrl{1} &[-0.5em] &[-1em] \coctrl{}\wire[d]{q}&[-0.5em] &[-0.5em] &[-1em] & &[-0.5em] &[-1.2em] &[-1em]&[-.5em]&[-1em]&\\
    & \gate[2]{P^{(1)}} &\gate[2]{D^{(1)}} & & \gate{R^{(1)}}& & & & & & &&&&\\
    & & &\gate{\Phi} &\coctrl{}\wire[u]{q}&\ctrl{1} & & \coctrl{}\wire[d]{q}& & & &&&&\\
    &  & &&\gate[2]{P^{(2)}} &\gate[2]{D^{(2)}}  & &\gate{R^{(2)}}& & & &&&&\\
    & && & &  &\gate{\Phi}&\coctrl{}\wire[u]{q}&\ctrl{1} & & \coctrl{}\wire[d]{q}& &&&\\
    \newwave[minimum height=3em]& &  && &&&&& &\defer{\hspace{-2.3em}\rotatebox{-18}{$\cdots$}} & \ctrl{1}&\defer{\hspace{-2.3em}\rotatebox{-18}{$\cdots$}}&\coctrl{}\wire[d]{q}&\\[.5em]
    & & && & & & & &\gate[2]{P^{(m)}} & &\gate[2]{D^{(m)}}&&\gate{R^{(m)}}\wire[u]{q} & \\
    & & & & & && & && & & \gate{\Phi} &\coctrl{}\wire[u]{q} &
\end{quantikz}}\\[2em]
&=\quad\autoname{\begin{quantikz}[wire types={q,b,q,b,q,n,b,q,b,q},classical gap=2pt, row sep={2em,between origins}]
    & & \ctrl{1} &[-.5em] &[-1em] &[-.5em] &[-.5em] &[-.75em] &[-.5em] &[-1em] \coctrl{}\wire[d]{q}&&\\
    & \gate[2]{P^{(1)}} &\gate[2]{D^{(1)}} & & & & & && \gate{R^{(1)}}&&\\
    & &  &\gate{\Phi}&\ctrl{1} & & & & &\coctrl{}\wire[u]{q}&\coctrl{}\wire[d]{q}&\\
    & \gate[2]{P^{(2)}} & & &\gate[2]{D^{(2)}} & & & & & &\gate{R^{(2)}}&\\
    & & & & & \gate{\Phi} & \ctrl{1}& & &\coctrl{}\wire[d]{q}&\coctrl{}\wire[u]{q}&\\
    \newwave[minimum height=2.5em]& \defer{\rotatebox{90}{$\cdots$}}& & & & & \defer{\hspace{2em}\rotatebox{-45}{$\cdots$}}& \ctrl{1}& &\defer{\hspace{3.5em}\rotatebox{90}{$\cdots$}}&\\[.2em]
    &\gate[2]{P^{(m)}}& & & & & &\gate[2]{D^{(m)}} & & &\gate{R^{(m)}}\wire[u]{q} & \\
    & & & & & & & &\gate{\Phi} & &\coctrl{}\wire[u]{q} &
\end{quantikz}}
\end{align*}

The $P^{(i)}$ gates each act on disjoint sets of qubits, and the $R^{(i)}$ gates commute and so can be organized into two layers, with each layer involving gates acting on disjoint sets of qubits (as illustrated above).
Standard upper bounds show that both these layers can be implemented in depth $\mc O(4^k)$~\cite{tucci1999rudimentary,vartiainen04}. Additionally, any $k$-qubit diagonal unitary can be implemented exactly by a $\mathsf{QNC}$ circuit using at most $\mc O(2^n)$ gates \cite{bullockOptDiag}, so the central cascade requires depth $\mc O(m 2^k)$. Adding these depths together completes the proof.
\end{proof}

\subsection{Stronger parallelization with ancillae}
\label{subsec:diagonal_caching}

This section focuses on the cascade of diagonal unitaries and single qubit rotations produced by~\Cref{thm:general-exact-Intro}. We show this cascade can be further parallelized if ancilla qubits are introduced. Unlike the previous subsections, the techniques used here are standard, though some care is taken in their application. The result of this parallelization is~\Cref{thm:general_parallelization-Intro}, which we restate now. 

\begin{reptheorem}{thm:general_parallelization-Intro}[Restated]
    For any $\mc C(\vec{U})$ with $m$-many $k$-qubit unitaries and error threshold~$\epsilon$ there exists:
    \begin{enumerate}[label = (\alph*)]
        \item A $\mathsf{QNC}$ circuit of depth $\mc O(4^k + mk)$ and with $m2^{k+1}$ ancilla qubits which implements $\mc C(\vec{U})$ exactly. 
        \item A $\mathsf{QNC}$ circuit of depth $\mc O(4^k + m\log(\log(m/\epsilon))$ and with $4m\log(m/\epsilon)$ ancilla qubits which implements a unitary $\mc C '$ satisfying $\|\mc C' - \mc C(\vec{U})\|_\infty \leq \epsilon$.
    \end{enumerate}
\end{reptheorem}

\begin{proof}

From the proof of~\Cref{thm:general-exact-Intro}, we know that $\mc C(\vec{U})$ can be written as a layer of pre- and postprocessing unitaries with depth $\mc O (4^k)$, surrounding a cascade of controlled-diagonal unitaries and single qubit rotations, with each diagonal unitary acting on $k$ qubits. To analyze this cascade, we first note that any diagonal gate $D^{(i)}$ can be written as a product of two multiplexer phase gates acting on a single qubit: 

\begin{align} \label{eq:diagonal_gate_id}
\autoname{\begin{quantikz}[wire types={q,b,q},classical gap=2pt, row sep={2em,between origins}]
&  \ctrl{1}             & \\
&  \gate[2]{D^{(i)}}    & \\
&  \ghost{D^{(i)}}      &
\end{quantikz}} 
= \autoname{\begin{quantikz}[wire types={q,b,q},classical gap=2pt, row sep={2em,between origins}]
&  \ctrl{1}                 &           & \ctrl{1}              &           &\\
&  \coctrl{}\wire[d]{q}     &           & \coctrl{}\wire[d]{q}  &           &\\
&  \gate{Z(\theta^{(i)})}   & \gate{X}  & \gate{Z(\phi^{(i)})}  & \gate{X}  & 
\end{quantikz}}.  
\end{align}

To formalize this notation, let $\theta^{(i)}, \phi^{(i)}$ both be vectors of length $2^{k-1}$, and let their components $\theta^{(i)}_x, \phi^{(i)}_x$ be indexed by a length $k-1$ bit strings $x$. Then the circuit indicates that the product of phase gates
\begin{align*}
    Z(\theta^{(i)}_x) X Z(\phi^{(i)}_x) X &:= \begin{pmatrix}
        1   & 0\\
        0                           & \exp(2\pi i \theta^{(i)}_x) 
    \end{pmatrix} X \begin{pmatrix}
        1   & 0\\
        0                           & \exp(2\pi i \phi^{(i)}_x) 
    \end{pmatrix}X \\
    &= \begin{pmatrix}
        \exp(2\pi i \phi^{(i)}_x)   & 0\\
        0                           & \exp(2\pi i \theta^{(i)}_x) 
    \end{pmatrix}
\end{align*} 
is implemented on the bottom qubit controlled on the $k-1$ qubits above it being in the computational basis state $\ket{x}$ and the upper qubit being in the state $\ket{1}$.\footnote{We use the notation $Z(\theta_x)$ instead of the more standard $P(\theta_x)$ for phase gates to avoid confusing with our preprocessing unitaries $P^{(i)}$.}

Now the controlled diagonal gates appearing in the central cascade in the proof of~\Cref{thm:general-exact-Intro} can be rewritten as products of multiplexer phase gates using~\Cref{eq:diagonal_gate_id}.
This produces a cascade of products of multiplexer phase gates and single qubit rotations. 
We will show how to approximately and exactly parallelize cascades of this form, proving claims (b) and (a) of~\Cref{thm:general_parallelization-Intro}, respectively. We begin with the approximate parallelization technique. \\

\noindent \textit{Proof of Item (b):} For notational convenience, we focus on just a single multiplexer phase gate, which we write as 
\begin{align*}
\autoname{\begin{quantikz}[wire types={q,b,q},classical gap=2pt, row sep={2em,between origins}]
&  \ctrl{1}                 & \\
&  \coctrl{}\wire[d]{q}     & \\
&  \gate{Z(\theta)} &
\end{quantikz}}.
\end{align*}
As above, this multiplexer phase gate implements one of $2^{k-1}$ phase gates $Z(\theta_x)$ on the final qubit, controlled on the computational basis state of the above qubits. Without loss of generality, we can assume that every $\theta_x \in [0,1)$. Then, for every $k$ bit string $x$, let $\hat{\theta}_x$ be the approximation to $\theta_x$ obtained by taking the first $r$ digits of its binary expansion. We write this as $\theta_x = 0.\theta_{x,1}\theta_{x,2} ... \theta_{x,r}$. 

Next, we define the unitary $\textsc{LOAD}(\hat{\theta})$ which acts on $k$ control qubits and $r$ ancilla qubits and loads a binary representation of $\hat{\theta}_x$ as 
\begin{align*}
    \textsc{LOAD}(\hat{\theta})  \ket{x}\ket{0}  = \ket{x}\ket{\theta_{x,1}\theta_{x,2}...\theta_{x,r}}.
\end{align*}
We write the inverse operation as $\textsc{UNLOAD}(\hat{\theta})$. In circuit diagrams, we indicate both $\textsc{LOAD}$ and $\textsc{UNLOAD}$ operations by multiplexer controls indicating the control qubits, and $\textsc{LOAD}/\textsc{UNLOAD}$ gates on the target (ancilla) qubits. 

Now we consider the following circuit, consisting of $k+1$ qubits and $2r$ ancillae: 
\begin{gather*}
\begin{quantikz}[wire types={q,b,q,q,q,n,q,q,q,n,q},classical gap=2pt, row sep={2em,between origins}]
\lstick[3]{$k+1$ qubits} & & \ctrl{10} & &[-0.5em] &[-1em] & & & \ctrl{10} & & \\
&  \coctrl{}\wire[d]{q}                         &           & & & & &                                   &           & \coctrl{}\wire[d]{q}                  & \\
&  \wire[d]{q}                                  & \phase{}  & & & & &                                   & \phase{}  & \wire[d]{q}                           &\\
\lstick[4]{$r$ ancillae}&  \gate[4]{\textsc{LOAD}(\hat{\theta}_x)}      &           & \ctrl{4} & & & &                          &           & \gate[4]{\textsc{UNLOAD}(\hat{\theta}_x)} &\\
&  \ghost{\textsc{LOAD}(\theta_x)}              &           & & \ctrl{4} & & &                          &           & \ghost{\textsc{LOAD}(\theta_x)}       & \\
&  \ghost{\textsc{LOAD}(\theta_x)}              &           & & & & \ddots{} &                          &           & \ghost{\textsc{LOAD}(\theta_x)}       & \\
&  \ghost{\textsc{LOAD}(\theta_x)}              &           & & & & & \ctrl{4}                          &           & \ghost{\textsc{LOAD}(\theta_x)}       & \\
\lstick[4]{$r$ ancillae}&                                               & \targ{}   & \gate{Z(1)} & & & &     & \targ{}   &                                       & \\
&                                               & \targ{}   & & \gate{Z(1/2)} & & &   & \targ{}   &                                       & \\
&                                               &           & & & & \ddots{} &                          &           &                                       & \\
&                                               & \targ{}   & & & & & \gate{Z(1/2^{r-1})} & \targ{}   &                                       &
\end{quantikz}.
\end{gather*}
Straightforward calculation of the ``phase-kickback'' (see, for example, Section 3.1 of \cite{cleve2000fast}) shows that this circuit implements a multiplexer $Z(\hat{\theta})$ operation on the first $k+1$ qubits. That is, the action of the circuit above is equivalent to acting with the circuit 
\begin{align*}
\autoname{\begin{quantikz}[wire types={q,b,q},classical gap=2pt, row sep={2em,between origins}]
&  \ctrl{1}                 & \\
&  \coctrl{}\wire[d]{q}     & \\
&  \gate{Z(\hat{\theta})} &
\end{quantikz}}
\end{align*}
on the first $k+1$ qubits. But for any $x$, we also have $|| Z(\theta_x) - Z(\hat{\theta}_x)||_\infty \leq 2^{-r}$ and so we also have that the $\textsc{LOAD}/\textsc{UNLOAD}$ circuit approximates our original multiplexer phase gate to error at most $2^{-r}$ in infinity norm. 

Now we consider the original cascade of controlled diagonal gates and single qubit rotations. We replace each controlled diagonal gate by multiplexer phase gates using~\Cref{eq:diagonal_gate_id}, then approximate each multiplexer phase gates using $2r$ ancillae and the technique described above.
The resulting \textsc{LOAD} operations all commute with each other and can be commuted to the front of the circuit and implemented in depth $\mc O (2^k)$.
Similarly the \textsc{UNLOAD} operations commute with each other and the single qubit rotations (they act on disjoint qubits from the single qubit rotations) and can be commuted to the end of the circuit and also implemented in depth $\mc O(2^k)$. 
What is left to implement is the cascade is $m$ iterations of the $r$-output Toffoli gates and two qubit controlled $Z$ rotations in the above figure. The Toffoli gates can each be implemented with depth $\log r$ using a standard binary-tree fanout, and the controlled $Z$ rotations can be implemented with constant depth. 

Then this method lets us approximate the cascade of controlled-diagonal unitaries and single qubit rotations in depth $\mc O(2^k + m\log r)$. Substituting this into the proof of~\Cref{thm:general-exact-Intro} gives an overall circuit depth of $\mc O(k^24^k + m\log r)$ which implements an approximate version of the original unitary cascade. The approximation error is $\mc O(m2^{-r})$ in infinity norm, since there are $2m$ multiplexed phase gates $Z(\theta)$ which were approximated by controlled phase gates $Z(\hat{\theta})$. We also require a total of $4mr$ ancillae to implement all the approximate phase gates. Setting $r = \log(m/\epsilon)$ completes the proof.\\

\noindent \textit{Proof of Item (a):} As in the proof of Item (b), we begin by considering a single multiplexer phase gate, which we write as 
\begin{align*}
\autoname{\begin{quantikz}[wire types={q,b,q},classical gap=2pt, row sep={2em,between origins}]
&  \ctrl{1}                 & \\
&  \coctrl{}\wire[d]{q}     & \\
&  \gate{Z(\theta)} &
\end{quantikz}}.
\end{align*}
Now we define the \textsc{SELECT} operation, which acts on $k-1$ qubits and $2^{k-1}$ ancillae by writing a $1$ on the ancilla bit indexed by the computational basis state of the qubits, that is 
\begin{align*}
    \textsc{SELECT} \ket{x}\ket{00...0} = \ket{x}\ket{\delta_{0,x}\delta_{1,x}...\delta_{2^{k-1}-1,x}}
\end{align*}
where $\delta$ is a Kronecker delta and we interpret the bit string $x$ as indexing an integer between $0$ and $2^{k-1}-1$. We also let \textsc{UNSELECT} be the inverse operation. As with the \textsc{LOAD}/\textsc{UNLOAD} operations, we indicate these operations in a circuit by multiplexer controls on the control qubits, and \textsc{SELECT}/\textsc{UNSELECT} gates on the ancillae. 

Now we consider the following circuit:
\begin{gather*}
\begin{quantikz}[wire types={q,b,q,q,q,n,q,q,q,n,q},classical gap=2pt, row sep={2em,between origins}]
\lstick[3]{$k+1$ qubits} & & \ctrl{10} & &[-0.5em] &[-1em] & & & \ctrl{10} & & \\
&  \coctrl{}\wire[d]{q}                         &           & & & & &                                   &           & \coctrl{}\wire[d]{q}                  & \\
&  \wire[d]{q}                                  & \phase{}  & & & & &                                   & \phase{}  & \wire[d]{q}                           &\\
\lstick[4]{$2^{k-1}$ ancillae}&  \gate[4]{\textsc{SELECT}}      &           & \ctrl{4} & & & &                          &           & \gate[4]{\textsc{UNSELECT}} &\\
&               &           & & \ctrl{4} & & &                          &           & \ghost{\textsc{LOAD}(\theta_x)}       & \\
&               &           & & & & \ddots{} &                          &           & \ghost{\textsc{LOAD}(\theta_x)}       & \\
&               &           & & & & & \ctrl{4}                          &           & \ghost{\textsc{LOAD}(\theta_x)}       & \\
\lstick[4]{$2^{k-1}$ ancillae}&                                               & \targ{}   & \gate{Z(\theta_0)} & & & &     & \targ{}   &                                       & \\
&                                               & \targ{}   & & \gate{Z(\theta_1)} & & &   & \targ{}   &                                       & \\
&                                               &           & & & & \ddots{} &                          &           &                                       & \\
&                                               & \targ{}   & & & & & \gate{Z(\theta_{2^{k-1}-1})} & \targ{}   &                                       &
\end{quantikz}.
\end{gather*}
where we have again identified $k-1$ bit strings with their integer representations in indexing the phase gates $Z(\theta_x)$. Direct calculation of the phase-kickback shows that this circuit exactly implements the desired multiplexer phase gate on the first $k+1$ qubits. 

Then, as in the approximate case, we can consider the original cascade of controlled diagonal unitaries and single qubit rotations, replace the diagonal gates by multiplexer phase gates, and then rewrite those gates using the technique above. The \textsc{SELECT} and \textsc{UNSELECT} gates can be commuted to the front and back of the circuit, respectively, and implemented in depth $\mc O(2^k)$~\cite{bullockOptDiag}. The $2^{k-1}$-output Toffoli and phase gates are implemented in a cascade, requiring total depth $\mc O(mk)$. So the total circuit depth required for this exact rewriting is $\mc O (4^k + mk) $ while requiring $m 2^{k+1}$ total ancillae.
\end{proof}

The reader will notice that the proof above assumes a worst-case $D^{(i)}$.
Given some promise about the structure of $D^{(i)}$'s eigenvalues---for example that there is a constant number of them, or that they are generated by a constant number of phases in the circle group---it would be possible to implement the $D^{(i)}$ gates exactly using fewer ancillae and a shallower circuit.
This observation is the starting point of the tighter Moore--Nilsson upper bound proven in~\Cref{sec:MN_tight_bounds}.

\section{The Moore--Nilsson conjecture}
\label{sec:MN_tight_bounds}

The Moore--Nilsson conjecture concerns circuits $\mc C(\vec U)$ where the controlled $U^{(j)}$'s are all single qubit unitaries.
This section specializes parallelization techniques from the previous section to this setting and proves tight upper bounds for these Moore--Nilsson circuits.
As we will explain, a key ingredient is a CS decomposition for a special class of circuits we term ``valley circuits.''
It turns out we will only need to use a small amount of the available theory for these circuits; see
\Cref{app:CSDecFlippedP} for a fuller picture of their structure.

\subsection{Partitioning Moore--Nilsson circuits}
\label{subsec:chunking}

As a first step towards parallelizing Moore--Nilsson circuits, we note that we can collect together cascades of controlled unitaries into groups as pictured in \Cref{fig:chunking-MN}. This lets us reinterpret the Moore--Nilsson cascade as a different cascade consisting of fewer multiplexers, each acting on a larger number of qubits.   

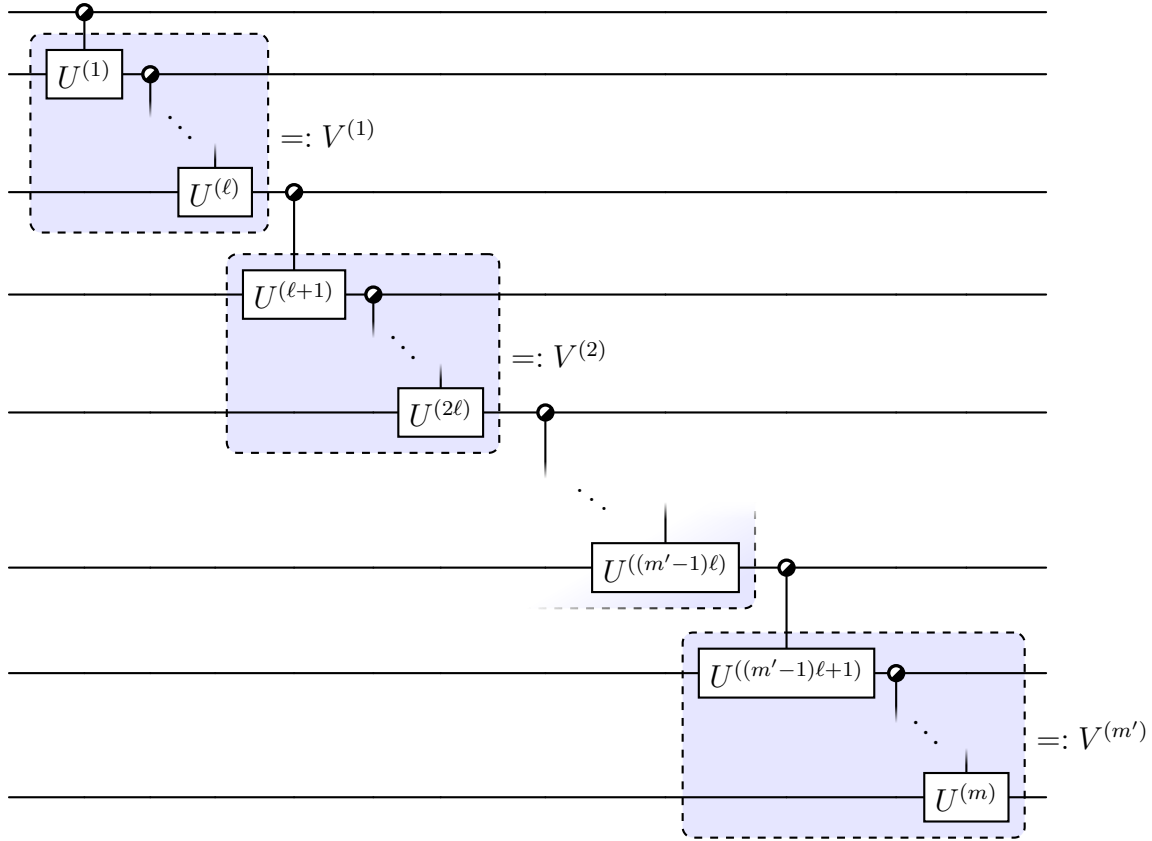
\begin{figure}[h]
    \centering
    
    \[
\autoname{\begin{quantikz}[wire types={q,q,n,q,q,n,q,n,q,q,n,q},classical gap=2pt, row sep={2em,between origins}]
    & \coctrl{}\wire[d]{q} &[-.6em] &[-.6em] &[-1.5em] &[-.6em] &[-.7em] &[.5em] & &[-2.5em] &[-.8em] &[-.6em] & \\
    & \gate{U^{(1)}} \gategroup[3,steps=3,style={inner
sep=2pt, fill=blue!10,dashed, rounded corners},background,label style={label
position=right,anchor=west,yshift=-0.2cm}]{$=: V^{(1)}$}             &    \coctrl{}\wire[d]{q}                                                                               &     
    &                               &                       &                       &                                       &                        
    &                               &                       &                       &                                                                   \\[-.2em]
    &           \newwavepartial[style={wave color=blue!10},steps=3]                    &                                              & \defer{\hspace{2em}\rotatebox{-45}{$\cdots$}}                               &  
                                  &                       &                       &                                       &                        
    &                               &                       &                       &                                                                \\
    &                               &                                              &                                       \gate{U^{(\ell )}}\wire[u]{q}
    & \coctrl{}\wire[d]{q}          &                       &                       &                                       &                        
    &                               &                       &                       &                                                                    \\[1.3em]
    &                               &                                              &                                         
    & \gate{U^{(\ell + 1)}} \gategroup[3,steps=3,style={inner
sep=2pt, fill=blue!10,dashed, rounded corners},background,label style={label
position=right,anchor=west,yshift=-0.2cm}]{$=: V^{(2)}$}         & \coctrl{}\wire[d]{q}                       &                       &                                                              
    &                               &                       &                       &                                       &                             \\[-.2em]
    &                                                                     &                                       & 
    &     \newwavepartial[style={wave color=blue!10},steps=3]                           &                        &            \defer{\hspace{2em}\rotatebox{-45}{$\cdots$}}                                                   
    &                               &                       &                       &                                       &                               \\
    &                               &                                              &                               
    &                                                     &                       & \gate{U^{(2\ell )}}\wire[u]{q}                      & \coctrl{}\wire[d]{q}      
    &                               &                       &                       &                                       &                               \\[.5em]
     &   \newwave[minimum height=2em]                    &                                              &                                       &                             
    &                               &                      &    \defer{\hspace{3em}\ddots}                         \gategroup[2,steps=2,style={inner
sep=2pt, fill=blue!10,path fading = bgfade,
  fading angle = 225,dashed, rounded corners},background,label style={label
position=right,anchor=west,yshift=-0.2cm}]{}    &                              &    
    &                              &                       &                       &                                       &                               \\[.5em]
%
%
    &                              &                       &                                                    
    &                               &                       &                       &                                                        
    & \gate{U^{((m'-1)\ell )}}\wire[u]{q}         & \coctrl{}\wire[d]{q}  &                       &                                       &                               \\[1.4em]
    &                               &                                              &                                       &                              
    &                               &                       &                       &                                       &                       
                \gate{U^{((m'-1)\ell + 1)}} \gategroup[3,steps=3,style={inner
sep=2pt, fill=blue!10,dashed, rounded corners},background,label style={label
position=right,anchor=west,yshift=-0.2cm}]{$=: V^{(m')}$}&   \coctrl{}\wire[d]{q}               &                                       &                           
\\
    &                                                                                                           &                              
    &                               &                       &                       &                                       &                       
    &     &                          \newwavepartial[style={wave color=blue!10},steps=3]                    &                                              & \defer{\hspace{2em}\rotatebox{-45}{$\cdots$}}                                                             & \\
    &                               &                                                                    
    &                               &                       &                       &                                       &                       
    &                               &                       &                       &                                       \gate{U^{(m)}}\wire[u]{q}  & 
\end{quantikz}}
\]
    \caption{Grouping subcircuits in a Moore--Nilsson circuit to obtain $\mc C(\vec V)$, a control cascade circuit of $m'$-many $j$-qubit unitaries.}
    \label{fig:chunking-MN}
\end{figure}

We pause here to remark that for appropriately sized groups, this partitioning of $\mc C(\vec U)$ is already sufficient to obtain mild depth reductions for Moore--Nilsson circuits. The key idea is to apply \Cref{thm:general_parallelization-Intro} to the each group as follows.

Each unitary in this cascade acts on $k' = \ell k \leq \log(n)/10$ qubits, and there are a total of $m' = m/\ell \leq 100n/\log n$ terms in the cascade.
Applying \Cref{thm:general_parallelization-Intro} with $\epsilon = 1/\mathrm{poly}(n)$ gives that this circuit can be implemented in depth 
\begin{align*}
    \mc O(4^{k'} + m' \log\log(m'/\epsilon)) \leq \mc O\left(n^{0.1} + \frac{100n \log\log(n \epsilon)}{\log n}\right) = \mc O\left(\frac{n\log\log n}{\log n}\right).
\end{align*}
This observation already gives a weak disproof of the Moore--Nilsson conjecture.
But we obtain much stronger parallelization results by exploiting the specific structure of the subcircuits formed by this grouping procedure.

\subsection{Valley circuits}
\label{subsec:flipped_pyramids}
Each of the groups $V^{(j)}$ discussed in the previous section take the form of smaller Moore--Nilsson cascades.
In this section we focus on just a single group, which we write as

\begin{align*}
    \autoname{\begin{quantikz}[wire types={q,q,q,n,q},classical gap=2pt, row sep = {1.75em, between origins}]
        & \coctrl{}\wire[d]{q} &\\
        & \gate[4]{V} &\\
        & & \\
        & & \\
        & & 
    \end{quantikz}}
    \quad=\quad
    \autoname{\begin{quantikz}[wire types={q,q,q,n,q},classical gap=2pt, row sep = {1.75em, between origins}]
        & \ctrl{1} & & & &\\
        & \gate{U^{(1)}} & \ctrl{1} & & & \\
        & & \gate{U^{(2)}} & \ctrl{1} & &\\
        \newwave &&&\defer{\hspace{2em}\rotatebox{-45}{$\cdots$}}& &\\
        & & & & \gate{U^{(\ell )}}\wire[u]{q} &
    \end{quantikz}}
\end{align*}
This multiplexer $V$  is the product of an open-controlled $V_0$ and closed-controlled $V_1$, where $V_0$ and $V_1$ are defined as:
\begin{gather*}
    \autoname{\begin{quantikz}[wire types={q,q,n,q},classical gap=2pt, row sep = {1.75em, between origins}, column sep = {0.5em}]
        & \gate[4]{V_0} &\\
        & & \\
        & & \\
        & & 
    \end{quantikz}}
    =
    \autoname{\begin{quantikz}[wire types={q,q,n,q},classical gap=2pt, row sep = {1.75em, between origins}, column sep = {0.5em}]
        & \ctrl{1} & & & \\
        & \gate{U^{(2)}} & \ctrl{1} & &\\
        \newwave &&\defer{\hspace{2em}\rotatebox{-45}{$\cdots$}}& &\\
        & & & \gate{U^{(\ell )}}\wire[u]{q} &
    \end{quantikz}}
    \quad \text{ and } \quad
    \autoname{\begin{quantikz}[wire types={q,q,n,q},classical gap=2pt, row sep = {1.75em, between origins}, column sep = {0.5em}]
        & \gate[4]{V_1} &\\
        & & \\
        & & \\
        & & 
    \end{quantikz}}
    =
    \autoname{\begin{quantikz}[wire types={q,q,n,q},classical gap=2pt, row sep = {1.75em, between origins}, column sep = {0.5em}]
        & \gate{U^{(1)}} & \ctrl{1} & & & \\
        & & \gate{U^{(2)}} & \ctrl{1} & &\\
        \newwave &&&\defer{\hspace{2em}\rotatebox{-45}{$\cdots$}}& &\\
        & & & & \gate{U^{(\ell )}}\wire[u]{q} &
    \end{quantikz}}
    .
\end{gather*}
In the previous section, we parallelized cascades by applying the quantum parallelization identity (\Cref{lem:q4rid-Intro}) to unitaries like the multiplexer $V$.
The proof of this identity involves applying the CS decomposition to the unitary $V_0^\dagger V_1$ with qubit order reversed which we write as $W = \rev (V_0^\dagger V_1)$.
But when $V$ has the form of a Moore--Nilsson cascade, the resulting unitary $W$ also takes on a specific form of a ``valley'' circuit:
\begin{align}
\label{eq:one_q_valley}
    W = 
    \autoname{\begin{quantikz}[wire types={q,q,n,q,q},classical gap=2pt, row sep = {1.75em, between origins}, column sep = {0.5em}]
        &\gate{{U^{(\ell )}}^\dagger} &&&& & &                       &        & \gate{U^{(\ell )}}    & \\
        &\ctrl{-1} &\gate{{U^{(\ell -1)}}^\dagger}\wire[d]{q} &&& & && \gate{U^{(\ell -1)}}\wire[d]{q}   &  \ctrl{-1}        & \\[0.3em]
        \newwave &&\defer{\hspace{2em}\rotatebox{-45}{$\cdots$}}&&&&&\defer{\hspace{1em}\rotatebox{45}{$\cdots$}\hspace{-1em}}&& &\\
        &&&\ctrl{-1}& \gate{{U^{(2)}}^\dagger} && \gate{U^{(2)}} & \ctrl{-1}&&& \\
        &&&&\ctrl{-1}& \gate{{U^{(1)}}} & \ctrl{-1} & & & &
    \end{quantikz}}.
\end{align}

The key technical result underpinning our stronger parallelization results is version of the cosine-sine decomposition adapted specifically to valley circuits. We state this next. 

\begin{lemma}[Cosine-sine for one-qubit valley circuits]
\label{lem:cos-sin_one_qubit}
    For any circuit $W$ of the form given in~\Cref{eq:one_q_valley}, we also have 
\[\mbox{\begin{quantikz}[wire types={q,b},classical gap=2pt, row sep={2em,between origins}]
&  \gate[2]{W} & \\
&  \ghost{W} & 
\end{quantikz}} 
= \mbox{\begin{quantikz}
[wire types 
= {q, b}, classical gap = 2pt, row sep={2em,between origins}]
& \coctrl{}\wire[d]{q} & \gate{\Phi^\dagger} & \gate{D} & \gate{\Phi} & 
 \coctrl{}\wire[d]{q} & \\
 & \gate{S} & & & & \gate{T} &  
\end{quantikz}}
\]
where $S, \Phi, T$ are defined as in~\Cref{cor:circuit_CS_decomp} and $D$ is a single-qubit diagonal unitary.  
\end{lemma}
Using this lemma in place of~\Cref{cor:circuit_CS_decomp} in the proof of~\Cref{lem:q4rid-Intro} gives the following quantum precomputation identity, first stated in the introduction.

\begin{replemma}{lem:q4rid-moore-nilsson-Intro}[Repeated]
Let $U^{(1)},\ldots, U^{(\ell)}$ be any $1$-qubit unitaries.
Then there exists an $(\ell-1)$-qubit unitary $P$, a 1-qubit unitary $Q$, and a doubly-controlled $(\ell-2)$-qubit unitary $R=(R_{00},R_{01},R_{10},R_{11})$
such that
\begin{equation}
\label{eq:q4rmn}
    \autoname{\begin{quantikz}[wire types={q,q,q,q},classical gap=2pt, row sep = {1.75em, between origins}]
        & \ctrl{1} & & & &\\
        & \gate{U^{(1)}} & \ctrl{1} & & & \\
        & & \gate{U^{(2)}} & \ctrl{1} & &\\
        \newwave &&&\defer{\hspace{2em}\rotatebox{-45}{$\cdots$}}& &\\
        & & & & \gate{U^{(\ell)}}\wire[u]{q} &
    \end{quantikz}}
    \quad=\quad
    \autoname{\begin{quantikz}[wire types={q,b,q},classical gap=2pt]
        & & \ctrl{2} &\coctrl{}\wire[d]{q} &\\
        & \gate[2]{P} & & \gate{R} &\\
        & & \gate{Q} & \coctrl{}\wire[u]{q} &
    \end{quantikz}}
\end{equation}
\end{replemma}

The next section proves more-general versions of \Cref{lem:cos-sin_one_qubit,lem:q4rid-moore-nilsson-Intro}, which apply to cascade and valley circuits consisting of multi-qubit controlled unitaries. This generality is not needed for the main result of this section (\Cref{thm:single_qbit_parallelization-Intro}), but is included because the more-general statements follow naturally from the proof techniques used. A proof of \Cref{thm:single_qbit_parallelization-Intro} using only \Cref{lem:q4rid-moore-nilsson-Intro} is given in \Cref{subsec:optimal_mn_circuits} and a reader may skip straight to this section if desired.

\subsection{A CS decomposition for valley circuits and a refined precomputation identity}
\label{subsec:flipped_pyramid_cosine_sine}

We begin by defining some notation for CS decompositions. 
For any even-dimensional unitary $Y$, let a CS decomposition be denoted
$$Y = \begin{pmatrix}
    S_0^Y & 0 \\ 0 & S_1^Y
\end{pmatrix}
\begin{pmatrix}
    \Sigma_1^Y & \Sigma_2^Y \\ -\Sigma_2^Y & \Sigma_1^Y
\end{pmatrix}
\begin{pmatrix}
    T_0^Y & 0 \\ 0 & T_1^Y
\end{pmatrix},$$
Recall that $\Sigma_1^Y, \Sigma_2^Y$ are diagonal matrices. We call these matrices the \textit{stubs} of  $Y$.
%
The \stubs \ $Y$ are not unique. We can permute the entries of $\Sigma_1^Y, \Sigma_2^Y$ simultaneously by conjugating with the permutation 
\begin{align*}
 I_2 \otimes \pi = \begin{pmatrix}
        \pi & 0 \\ 0 & \pi
\end{pmatrix},
\end{align*} 
which can be absorbed into $S$ and $T$.
Similarly, we can change the phase of the entries of $\Sigma_1^Y, \Sigma_2^Y$ be left multiplying by a diagonal matrix, and then absorbing that matrix into the $S_0^Y, S_1^Y$ matrices. 
Throughout this section we use the convention that the stubs are nonnegative diagonal matrices but entries can be arranged in any order.

The \stubs of a valley circuit can be written in terms of the \stubs the unitaries which define the valley circuit.
The next lemma gives such
an algebraic stub formula. In \Cref{scc:CSDecFlippedP} we expand this result further and give formulas 
for all of the components $S,T$ and \stubs the 
CS Decomposition.
The appendix proof is also stated in terms of circuit diagrams.

\begin{lemma}
\label{lem:SinglePyramidDelta}
Let unitaries $U$, $V$ and $W$ be related as in the following circuit. 
\begin{align*}
    W = \begin{quantikz}[classical gap=0.07cm, wire types = {b, q, b}]
        & \gate{U^\dagger} & & \gate{U} & \\
        & \ctrl{-1} & \gate[2]{V} & \ctrl{-1} &\\
        & & \ghost{V} & &
     \end{quantikz}
\end{align*}
Then, for any CS Decompositions for $U$, $V$, and $W$, defining 
$$M := (S_0^{U} \otimes \ketbra00 \otimes T_0^{V\dagger} + T_0^{U\dagger} \otimes \ketbra11 \otimes T_1^{V\dagger}) \; T_0^{W\dagger}.$$
gives
$$(\Sigma_2^W)^2 = M^\dagger \; ( (\Sigma_2^U)^2 \otimes I_2  \otimes (\Sigma_2^V)^2) \; M.$$
Consequently, there exists a CS Decomposition of $W$ with $T_0^W$ chosen to make $M = I$ and
$$(\Sigma_2^W)^2 = (\Sigma_2^U)^2 \otimes I_2 \otimes (\Sigma_2^V)^2.$$

\end{lemma}

\begin{proof}
To start, we can break $W$ into its blocks:
    $$W = \begin{pmatrix}
        W_{00} & W_{01} \\ W_{10} & W_{11}
    \end{pmatrix}$$
    
    Next, we want to rewrite $W_{10}$ in terms of $U$ and $V$. First, using $\ket0$ and $\ket1$ we can use the middle wire to split the blocks to rewrite~$W$:
    $$
    W = [(U^\dagger \ \otimes \ \ketbra11 \ + \ I_{\dim(U)} \ \otimes \ketbra00 )\ \otimes \ I_{\dim (V)/2} ]\ (I_{\dim(U)} \otimes V )\  [(U \otimes \ \ketbra11 \ + \ I_{\dim(U)} \ \otimes \ \ketbra00 ) \ \otimes \ I_{\dim(V)/2}].
    $$
    
    Breaking $V$ into its blocks
    $$V = \begin{pmatrix}
        V_{00} & V_{01} \\ V_{10} & V_{11}
    \end{pmatrix},$$
    we can then rewrite
    $$V = \ketbra00 \otimes V_{00} + \ketbra10 \otimes V_{10} + \ketbra01 \otimes V_{01} + \ketbra11 \otimes V_{11}.$$
    With this, we can multiply out $W$ to get
    $$
    W = I_{\dim(U)} \otimes \ketbra00 \otimes V_{00} \; + \; U^\dagger \otimes \ketbra10 \otimes V_{10} \; + \; U \otimes \ketbra01 \otimes V_{01} + I_{\dim(U)} \otimes \ketbra11 \otimes V_{11}.
    $$
    From here, we can express $W_{10}$ using $W$:
    \begin{align*}
        W_{10} &= ( \bra1 \otimes I_{\dim(W)/2} )\ W\ (\ket0 \otimes I_{\dim(W)/2}) \\
        &= (\bra1 \otimes I_{\dim(U)/2} \otimes I_2 \otimes I_{\dim(V)/2})\ W\ (\ket0 \otimes I_{\dim(U)} \otimes I_2 \otimes I_{\dim(V)/2})\\
        &= [(\bra1 \otimes I_{\dim(U)/2})\ U^\dagger \ (\ket0 \otimes I_{\dim(U)/2})] \otimes V_{10}\ + \ [(\bra1 \otimes I_{\dim(U)})\ U \ (\ket0 \otimes I_{\dim(U)})] \otimes V_{01}.
    \end{align*}
   To continue we expand $U$ similarly to $W$ and $V$ as
    $$U = \begin{pmatrix}
        U_{00} & U_{01} \\ U_{10} & U_{11}
    \end{pmatrix}.$$
    With this,
    $$
    W_{10} = U_{01}^\dagger \otimes \ketbra10 \otimes V_{10} \; + \; U_{10} \otimes \ketbra01 \otimes V_{01},
    $$
    so
    \begin{align*}
        W_{10}^\dagger W_{10} 
        &= U_{01} U_{01}^\dagger \otimes \ketbra00 \otimes V_{10}^\dagger V_{10} \ + \ U_{10}^\dagger U_{10} \otimes \ketbra11 \otimes V_{01}^\dagger V_{01}.
    \end{align*}

    Now we write a CS decomposition of $U$ as
    $$
    U = \begin{pmatrix}
        S_0^U & 0 \\ 0 & S_1^U
    \end{pmatrix}\begin{pmatrix}
        \Sigma_1^U & \Sigma_2^U \\ -\Sigma_2^U & \Sigma_1^U
    \end{pmatrix}\begin{pmatrix}
        T_0^U & 0 \\ 0 & T_1^U
    \end{pmatrix}
    $$
    so $U_{10} = -S_1^U\Sigma_2^UT_0^U$ and $U_{01} = S_0^U \Sigma_2^UT_1^U$. Likewise, we can write a CS Decomposition of $V$ and $W$ so $V_{10} = -S_1^V\Sigma_2^VT_0^V$, $V_{01} = S_0^V \Sigma_2^VT_1^V$, and $W_{10} = -S_1^W\Sigma_2^WT_0^W$.
    With this,
    $$W_{10}^\dagger W_{10} = S_0^{U}\Sigma_2^{U}\Sigma_2^{U\dagger}S_0^{U\dagger} \otimes \ketbra00 \otimes T_0^{V\dagger}\Sigma_2^{V\dagger}\Sigma_2^VT_0^V + 
    T_0^{U\dagger} \Sigma_2^{U\dagger} \Sigma_2^{U}T_0^{U} \otimes \ketbra11 \otimes T_1^{V\dagger}\Sigma_2^{V\dagger}\Sigma_2^VT_1^V.$$
    
    Conjugating $W_{10}^\dagger W_{10}$ by the unitary matrix 
    $S_0^{U} \otimes \ketbra 00 \otimes T_0^{V\dagger} + T_0^{U\dagger} \otimes \ketbra 11 \otimes T_1^{V\dagger}$ 
    we find
    \begin{align*}
        \Big(S_0^{U} \otimes \ketbra 00 \otimes T_0^{V\dagger} + T_0^{U\dagger} &\otimes \ketbra 11 \otimes T_1^{V\dagger}\Big)^\dagger \ W_{10}^\dagger W_{10} \  \Big(S_0^{U} \otimes \ketbra 00 \otimes T_0^{V\dagger} + T_0^{U\dagger} \otimes \ketbra 11 \otimes T_1^{V\dagger}\Big) \\[0.5em]
        &= \Sigma_2^{U\dagger}\Sigma_2^U \otimes \ketbra00 \otimes \Sigma_2^{V\dagger}\Sigma_2^{V} + \Sigma_2^{U}\Sigma_2^{U\dagger} \otimes \ketbra11 \otimes \Sigma_2^{V\dagger}\Sigma_2^{V} \\[0.5em]
        &=\Sigma_2^{U\dagger}\Sigma_2^U \otimes I_2 \otimes \Sigma_2^{V\dagger}\Sigma_2^V
    \end{align*}
    At the same time, $W_{10}^\dagger W_{10} = T_0^{W\dagger}\Sigma_2^{W\dagger}\Sigma_2^WT_0^W$ from the CS Decomposition for $W$.

Then we set 
\begin{align*}
M
:=
  \Big(S_0^{U} \otimes \ketbra 00 \otimes T_0^{V\dagger} \ + \ T_0^{U\dagger} \otimes \ketbra 11 \otimes T_1^{V\dagger}\Big) \; T_0^{W\dagger}
\end{align*}
so
\begin{align*}
    \Sigma_2^{W\dagger}\Sigma_2^W = M^\dagger \; \Big(\Sigma_2^{U\dagger}\Sigma_2^U \otimes I_2 \otimes \Sigma_2^{V\dagger}\Sigma_2^V\Big) \; M.
\end{align*}

Note to this point we have not used that
the $\Sigma_2$
are real diagonal matrices.
But since they come from the CS Decomposition and are therefore real we have
$$(\Sigma_2^W)^2 = M^\dagger \ \Big((\Sigma_2^U)^2 \otimes I_2 \otimes (\Sigma_2^V)^2\Big) \ M,$$
as desired.
\end{proof}

Applying the above lemma inductively gives us our desired CS decomposition for valleys. The resulting lemma applies to a larger class of valley circuits in which each of the unitaries $U^{(i)}$ can possibly acting on many qubits, although the single-qubit case remains the strongest version of the statement. We define this class of circuits first, then state the lemma and a useful corollary. 

\begin{definition}
    Let $\vec{U} = (U^{(1)}, U^{(2)}, ..., U^{(\ell )})$ be a vector of unitaries, with the number of qubits each unitary acts on not necessarily equal. Then the valley circuit $\mc V (\vec{U})$ is the unitary given by 
    \begin{align*}
    \mc V(\vec{U}) = 
    \autoname{\begin{quantikz}[wire types={b,q,b,n,q,b,q,b},classical gap=2pt, row sep = {1.75em, between origins}, column sep = {0.5em}]
        &\gate{{U^{(\ell )\dagger}}} &&&& & &                       &        & \gate{U^{(\ell )}}    & \\
        &\ctrl{-1} &\gate[2]{{U^{(\ell -1)\dagger}}}\wire[d]{q} &&& & && \gate[2]{U^{(\ell -1)}}\wire[d]{q}   &  \ctrl{-1}        & \\
        & &\wire[d]{q} &&& & && \wire[d]{q}   &  & \\[0.3em]
        \newwave &&\defer{\hspace{2em}\rotatebox{-45}{$\cdots$}}&&&&&\defer{\hspace{1em}\rotatebox{45}{$\cdots$}\hspace{-1em}}&& &\\
        &&&\ctrl{-1}& \gate[2]{U^{(2)\dagger}} && \gate[2]{U^{(2)}} & \ctrl{-1}&&& \\
        &&&&  && & &&& \\
        &&&&\ctrl{-1}& \gate[2]{U^{(1)}} & \ctrl{-1} & & & & \\
        &&&&& & & & & &
    \end{quantikz}}
    \end{align*}
\end{definition}

\begin{lemma}
\label{lem:fpyrDelta}
    For any valley $\mc V (U^{(1)}, U^{(2)}, ..., ,U^{(\ell )})$, and arbitrary CS decompositions of $U^{(1)}, U^{(2)}, ..., ,U^{(\ell )}$, there exists a CS decomposition of $\mc V$ with
\begin{align*}
    (\Sigma_2^{\mc V})^2
  = (\Sigma_2^{U^{(\ell )}})^2 \otimes
 \bigotimes_{j=1}^{\ell-1}
  [   (\Sigma_2^{U^{(\ell -j)}})^2 \otimes I_2 ] . 
\end{align*}

\end{lemma}

\begin{proof} 
The proof is by induction on $\ell$. For $\ell = 1$, the identity is trivial. Let $\mc V ' = \mc V(U^{(1)}, U^{(2)}, ..., ,U^{(\ell -1)})$ be a valley with the first and last controlled $U^{(\ell )}$ removed. Then, by~\Cref{lem:SinglePyramidDelta}, for any CS decompositions of $\mc V'$ and $U^{(\ell )}$ there exists a CS decomposition of $\mc F$ with
\begin{align*}
    (\Sigma_2 ^ {\mc V})^2 = (\Sigma_2^{{U^{(\ell )}}})^2 \otimes I_2 \otimes (\Sigma_2 ^{\mc V '})^2.
\end{align*}
And, by our induction hypothesis we have that there exists a CS decomposition of $\mc V'$ with
\begin{align*}
    (\Sigma_2^{\mc V'})^2
  = 
 (\Sigma_2^{U^{(\ell - 1)}})^2 \otimes
 \bigotimes_{j=2}^{\ell-1}
  \Big(   (\Sigma_2^{U^{(\ell - j)}})^2 \otimes I_2 \Big) 
\end{align*}
Inserting this expression into the one above completes the proof. 
\end{proof}

\begin{corollary}
\label{prop:1qbflippedDelta}
For a valley $\mc V (U^{(1)}, U^{(2)}, ..., U^{(\ell )})$: 
\begin{enumerate}
\item \label{it:prop1qbflippedDelta1} 
If $U^{(1)}, U^{(2)}, ..., U^{(\ell )}$ are all one qubit unitaries, there exists a CS decomposition of $\mc V$ where $\Sigma_2^{\mc V}$ is a scalar multiple of the identity.

\item \label{it:prop1qbflippedDelta2} 
For any $U^{(1)}, U^{(2)}, ... , U^{(\ell )}$ there exists CS Decompositions of $\mc V$ with
\begin{align*}
    \Sigma_2^{\mc V} = I_{2^{\ell - 1}} \otimes \bigotimes_{j = 1}
^\ell \Sigma_2^{U^{(j)}} \qquad \text{or with} \qquad \Sigma_2^{\mc V} = \left(\bigotimes_{j = 1}
^\ell \Sigma_2^{U^{(j)}} \right) \otimes I_{2^{\ell-1}}.
\end{align*}
\end{enumerate}
\end{corollary}
\begin{proof}
\Cref{it:prop1qbflippedDelta2} follows immediately from \Cref{lem:fpyrDelta} and non-uniqueness of CS Decompositions under block permutations of the $\Sigma$'s. \Cref{it:prop1qbflippedDelta1} follows from \Cref{it:prop1qbflippedDelta2} with all $\Sigma_2^{U^{(j)}}$ being $1 \times 1$ matrices and therefore scalars. 
\end{proof}

\label{subsec:proof_of_MNQ4R}

\def\tU{{\widetilde U}}

We can now prove a refined version of our quantum precomputation identity (\Cref{lem:q4rid-Intro}), which bounds the size of the diagonal matrix $D$ in the circuit. As in the case of~\Cref{prop:1qbflippedDelta}, this lemma is stated for cascades with arbitrary-sized controlled unitaries, but the strongest versions of the statement apply to one-qubit cascades. 

\begin{lemma}
    [Quantum precomputation identity with stub count]
\label{prop:q4rid-stubs moore-nilsson}

Let $\mc C (U^{(1)}_0,U^{(1)}_1,\ldots, U^{(\ell)}_0,U^{(\ell)}_1)$ be a control cascade which acts on a total of $d$ qubits, with unitaries $U^{(1)}_0,U^{(1)}_1,\ldots, U^{(\ell)}_0,U^{(\ell)}_1$ arbitrary.
Then there exists an $(d-1)$-qubit unitary $P'$, a diagonal $(d - \ell)$-qubit unitary $D'$, and a multiplexer $R$
such that
\begin{equation}
\label{eq:q4rmn}
    \mathcal C(\vec U)
    \quad=\quad
    \begin{quantikz}[wire types={q,b,b,q},classical gap=2pt]
        & & & \ctrl{2} & & \coctrl{}\wire[d]{q} &\\
        & \gate[3]{P'} & & & & \gate[2]{R} &\\
        & \ghost{P'} & & \gate[2]{D'} & & \ghost{R} &\\
        & \ghost{P'} & \gate{\Phi^\dagger} & \ghost{D'} & \gate{\Phi} & \coctrl{}\wire[u]{q} &
    \end{quantikz}\,.
\end{equation}
Formulas for all these unitaries are given explicitly given in the proof.
\end{lemma}

\begin{proof}
This result is a refinement of the quantum precomputation identity 
$$ \mc C (U^{(1)}_0,U^{(1)}_1,\ldots, U^{(\ell)}_0,U^{(\ell)}_1) =
\begin{quantikz}[wire types={q,b,q},classical gap=2pt]
        & & & \ctrl{2} & & \coctrl{}\wire[d]{q} &\\
        & \gate[3]{P'} & & \gate[3]{D} & & \gate[2]{R} &\\
        & \ghost{P'} & \gate{\Phi^\dagger} & \ghost{D} & \gate{\Phi} & \coctrl{}\wire[u]{q} &
    \end{quantikz}
$$
We  will show the diagonal matrix $D$ in the quantum precomputation identity factors as
\begin{align*}
D=I^{\otimes \ell-1}_2 \otimes D'
\end{align*}
for a $(d - \ell)$-qubit diagonal unitary $D'$.

The $D$ in the quantum precomputation identity is simply the $D$ coming from the CS decomposition of valley with multiplexers $U^{(\ell)}, ..., U^{(1)}$. But this valley of multiplexers is equivalent to a valley of singly controlled unitaries with the same dimension, since we can rewrite a multiplexer valley via repeated applications of the identity 
\begin{align*}
    \begin{quantikz}[classical gap=0.07cm, wire types = {b, q, b}, column sep = 3pt]
            & \gate{U_0} & \gate{U_1} & & \gate{U_1^\dagger} & \gate{U_0^\dagger} & \\
            & \octrl{-1} & \ctrl{-1} & \gate[2]{V} & \ctrl{-1} & \octrl{-1} &\\
            & & & \ghost{V} & & &
         \end{quantikz}
         &=
         \begin{quantikz}[classical gap=0.07cm, wire types = {b, q, b}, column sep = 3pt]
            & \gate{U_1U_0^\dagger} & \gate{U_0} & & \gate{U_0^\dagger} & \gate{U_0U_1^\dagger} & \\
            & \ctrl{-1} & & \gate[2]{V} & & \ctrl{-1} &\\
            & & & \ghost{V} & & &
         \end{quantikz} \\
         &=
         \begin{quantikz}[classical gap=0.07cm, wire types = {b, q, b}, column sep = 3pt]
            & \gate{U_1 U_0^\dagger} & & \gate{U_0 U_1^\dagger} & \\
            & \ctrl{-1} & \gate[2]{V} & \ctrl{-1} &\\
            & & \ghost{V} & &
         \end{quantikz}.
\end{align*}

Then by \Cref{lem:SinglePyramidDelta}, we can select a cosine-sine decomposition with the stub 
$$\Sigma_2 = I^{\otimes \ell - 1}_2 \otimes \bigotimes_{j = 1}^\ell \Sigma_2^{U^{(j )}}.$$
We write this as 
$$\Sigma_2 = I_2^{\otimes\ell - 1} \otimes \widehat \Sigma_2$$
where $\widehat \Sigma_2 = \bigotimes_{j = 1}^\ell \Sigma_2^{U^{(j)}}$.
For the same cosine-sine decomposition we have
$\Sigma_1$ by 
$$\Sigma_1 = \sqrt{I^{\otimes d - 2}_2 - \Sigma_2^2}
= \sqrt{I^{\otimes d - 2}_2 - I^{\otimes \ell - 1}_2 \otimes  \widehat \Sigma_2^2}
= I^{\otimes \ell - 1}_2 \otimes \widehat \Sigma_1$$
where
$$\hat \Sigma_1 = \sqrt{I^{\otimes d - \ell - 1}_2 - \widehat \Sigma_2^2}.$$
From $D = \begin{pmatrix}
    \Sigma_1 - i \Sigma_2 & 0 \\ 0 & \Sigma_1 + i \Sigma_2
\end{pmatrix}$, we can factor $I^{\otimes \ell - 1}_2$ to get
$$D = \begin{pmatrix}
    I^{\otimes \ell - 1}_2 \otimes (\widehat \Sigma_1 - i \widehat\Sigma_2) & 0 \\
    0 & I^{\otimes \ell - 1}_2 \otimes (\widehat \Sigma_1 + i \widehat\Sigma_2)
\end{pmatrix}
= I^{\otimes \ell - 1}_2 \otimes D'
$$
where 
$$
D' :=  
\begin{pmatrix}
    \widehat \Sigma_1 - i \widehat\Sigma_2 & 0 \\
    0 & \widehat \Sigma_1 + i \widehat\Sigma_2
\end{pmatrix}. 
$$
\end{proof}

\subsection{Optimal-depth Moore--Nilsson circuits}
\label{subsec:optimal_mn_circuits}

We are now ready to prove our optimal depth compression result for Moore--Nilsson circuits (\Cref{thm:single_qbit_parallelization-Intro}). We begin by repeating the theorem here, then give its proof.

\begin{reptheorem}{thm:single_qbit_parallelization-Intro}[Repeated]
    Every Moore--Nilsson unitary $\mc C(\vec U)$ on $n$ qubits is computed exactly by...
    \begin{itemize}
        \item A $\textsf{QNC}$ circuit of depth $\mc O(\log n)$ and no ancillae, and
        \item A $\textsf{QNC}_\textsf{2D}$ circuit of depth $\mc O(\sqrt n)$ and $\mc O(n)$ ancillae.
    \end{itemize}
    Both of these depths are the best possible.
    Moreover, these circuits can be computed from the list of gates $U^{(1)},\ldots, U^{(n)}$ in time $\mathrm{poly}(n)$.
\end{reptheorem}

\begin{proof}
Group the cascade into blocks $V^{(j)}$ of size $\ell = b$ as pictured in \Cref{fig:chunking-MN}, apply the quantum precomputation identity for Moore--Nilsson circuits (\Cref{lem:q4rid-moore-nilsson-Intro}), and rearrange to obtain:

\[\mc C(\vec U)=\autoname{\begin{quantikz}[wire types={q,b,q,b,q,n,b,q,b,q},classical gap=2pt, row sep={2em,between origins}, column sep = {2em, between origins}]
    & &[2em] \ctrl{2} \gategroup[8,steps=4,style={inner
sep=2pt, fill=blue!10,dashed, rounded corners},background,label style={label
position=south,anchor=north,yshift=-0.5em}]{$=\mc C(\vec Q)$}& & & & &\coctrl{}\wire[d]{q}&& \\
    & \gate[2]{P^{(1)}} & & & & & &\gate{R^{(1)}}&&\\
    & & \gate{Q^{(1)}} &\ctrl{2} & & & &\coctrl{}\wire[u]{q}&\coctrl{}\wire[d]{q}&\\
    & \gate[2]{P^{(2)}} &  & & & & &&\gate{R^{(2)}}&\\
    & & & \gate{Q^{(2)}}& \ctrl{1}& & &\coctrl{}\wire[d]{q}&\coctrl{}\wire[u]{q}&\\
    & \defer{\rotatebox{90}{$\cdots$}}& \newwavepartial[style={minimum height=2.5em, wave color=blue!10},steps=4]{} && \defer{\hspace{2em}\rotatebox{-45}{$\cdots$}}& \ctrl{2}& \newwavepartial[style={minimum height=2.5em, wave color=white},steps=4]{}&\defer{\hspace{2em}\rotatebox{90}{$\cdots$}}&\\[.2em]
    & \gate[2]{P^{(m)}}& & & & & & &\gate{R^{(m)}}\wire[u]{q} & \\
    & & & & & \gate{Q^{(m)}}& & &\coctrl{}\wire[u]{q} &
\end{quantikz}}
\]

The depth of the first column is at most $\mc O(4^b)$ \cite{tucci1999rudimentary,vartiainen04}; same for the last column.
The middle column is a Moore--Nilsson circuit on $\mc O\big(n/b\big)$ qubits.
Iterating this identity on successive $\mc C(\vec Q)$'s $r$-many times yields a circuit of total depth at most
\[\mc O\left(r 4^b+\frac{n}{b^r}\right)\,.\]
Setting $b$ to a sufficiently large constant and $r=\log(n)$ we find the depth is at most~$\mc O(\log n)$ for all-to-all connected $\cc{QNC}$ circuits.
This completes the first assertion of \Cref{thm:single_qbit_parallelization-Intro}.

\bigskip

We turn to the second assertion of the theorem and describe how to compile the 
circuit above into a 2D connectivity architecture with optimal depth.
We will work with the concrete choice of block size $b = 2$.

Given any quantum circuit, we call its \emph{forward topology} the directed graph obtained by, for each qubit $j$, truncating the $j$\textsuperscript{th} wire so that it terminates at the final gate that interacts with qubit $j$.
Vertices in the resulting graph correspond to gates in the circuit, while edges in the graph correspond to wires and are oriented with the flow of time.
Parallel edges (i.e. edges with the same out and in vertex) are replaced with a single edge and a label denoting the number of qubits represented by the edge.
The \emph{backward topology} is defined analogously, with time reversed.
This is illustrated in the example of a $9$-qubit Moore--Nilsson circuit in \Cref{fig:2dtopo} and  \Cref{fig:Rtopo}.
Now we argue that the first and last ``stages'' of the circuit have nice forwards and backwards topologies, respectively.
\\

\begin{figure}
    \centering
    \includegraphics[width=0.5\linewidth]{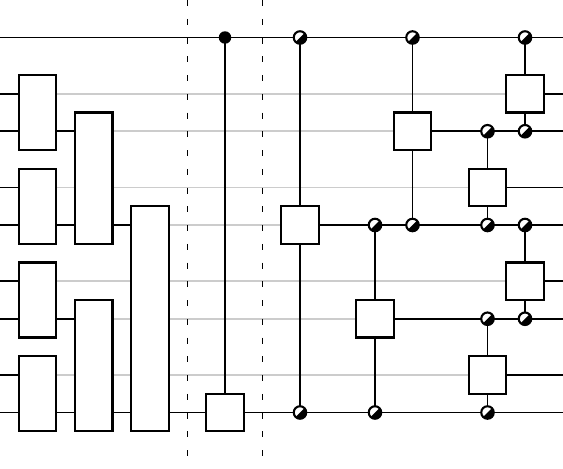}
    \caption{An example $9$-qubit Moore--Nilsson circuit after parallelization.
    It naturally separates into three stages, with the first stage having the forward topology of a complete binary tree.}
    \label{fig:2dtopo}
\end{figure}

\begin{claim}
     The $P$ stage of the Moore--Nilsson parallelization has forward topology equal to the complete binary tree.
     The $R$ stage of the Moore--Nilsson parallelization can be rewritten into a circuit with backward topology equal to the complete binary tree, followed by a layer of CNOTs.
\end{claim}
\begin{claimproof}
    The first part is immediate.
    The second part can be seen by splitting shared control wires into left and right sides with CNOTs that can be easily uncomputed in depth one afterwards.
    See \Cref{fig:Rtopo}.
\end{claimproof}

\begin{claim}
    Fix an embedding of the $n$-leaf complete binary tree in an $m$-vertex 2D grid and let $d$ be the maximum edge length in the embedding.
    Then any quantum circuit with complete binary tree forward topology and where each gate is size $\mc O(1)$ and each edge is on $\mc O(1)$ qubits has a compilation into the 2D architecture on $\mc O(m)$ qubits with depth $\mc O(d\log n)$.
\end{claim}
\begin{claimproof}
    We describe a procedure for deriving a 2D quantum circuit from the tree embedding.

    We begin by identifying our $m$-vertex grid with a course-grained version of the 2D architecture of qubits, so that each vertex of the grid corresponds to a constant-sized square of qubits. We set this constant at least as large as the maximum number of qubits involved in a single gate in our original circuit.
    
    Now the gates corresponding to each level of the tree will be computed in parallel as a different stage (group of consecutive layers) of the quantum circuit.
    Each gate is implemented on the subset of qubits identified by the embedding. 
    Between the gates corresponding to vertices at level $s$ and level $s+1$ of the tree there is a swap network permuting qubits to they are in the correct location for the next layer of gates. This swap network permutes qubits along the paths corresponding to embedded tree edges.
    The depth of this swap network is bounded by a constant factor times the maximum embedded edge length among the level $s$-$(s+1)$ edges.
    The $n$-leaf binary tree has depth $\log n$, so the swap networks in total contribute $\mc O(d\log n)$ depth.
    The node layers in total contribute $\log n \cdot \mc O(1) = \mc O(\log n)$ depth, which is dominated by the swap network depth.
\end{claimproof}

\begin{figure}
    \centering
    \includegraphics[width=1\linewidth]{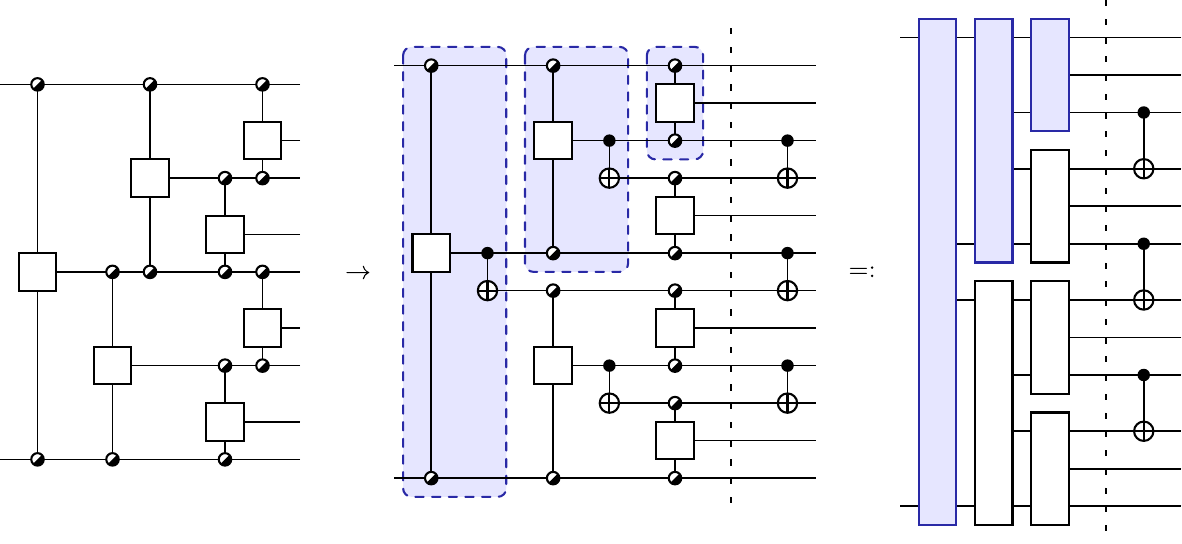}
    \caption{Determining the backward topology of the ``$R$ stage'' of the compressed Moore--Nilsson unitary.
    After some rewriting we obtain a binary tree with 2-qubit edges followed by a layer of CNOTs.}
    \label{fig:Rtopo}
\end{figure}

Combining these claims, we may appeal to minimax edge length embeddings of binary trees in 2D grids from the VLSI design literature \cite{patersonRuzzoSnyder,Ruzzo1981}.
These results state a binary tree with $n$ leaves may be embedded in a 2D grid with maximum edge length $\mc O(\sqrt{n}/\log n)$.
This leads to a total depth of $\mc O(\sqrt{n})$ for implementing the first and last stages of a parallelized Moore--Nilsson circuit with a 2D architecture. 

The central stage of the parallelized Moore--Nilsson circuit can be implemented in constant depth provided qubits on the 2D grid are arranged correctly. 

Connecting these three stages requires at worst two arbitrary permutations on the grid of qubits, which take depth $\mc O(\sqrt{n})$~\cite{alon1993routing}. Then the overall depth required to implement the parallelized Moore--Nilsson circuit is also $\mc O(\sqrt{n})$.

\bigskip
A straightforward lightcone argument shows that there are Moore--Nilsson unitaries that require depth $\Omega(\log n)$ in all-to-all connected circuits and depth $\Omega(\sqrt n)$ in 2D circuits, so these upper bounds are asymptotically tight. 
\end{proof}

We close this section by noting that our proof of the second half of~\Cref{thm:single_qbit_parallelization-Intro} was somewhat wasteful in its use of ancilla qubits. With more careful bookkeeping it may be possible to obtain a similar result using only $n + o(n)$ ancillae, but we leave this question to future work. 

\section{Discussion}
\label{sec:discussion}

This section covers two directions in which the results of this paper could be extended. In~\Cref{subsec:open_problems} we discuss slight generalizations of Moore--Nilsson circuits on which our preprocessing techniques do not immediately apply. Further study of these circuits may lead to either new parallelization results, or super-logarithmic depth lower bounds. In~\Cref{subsec:compilation} we discuss how the results in this paper may be converted to practically-relevant compilation techniques, and highlight some initial progress in this direction given in~\Cref{app:CSDecFlippedP}.

\subsection{Open problems near the Moore--Nilsson conjecture}

\label{subsec:open_problems}

Despite the progress made in this paper, we don't know how to beat the naive depth bounds in any of the situations described below.
More specifically, we don't know how to compile the unitaries below to smaller depth than their naive implementations, either exactly or while retaining small error in operator norm.
Interestingly, log-depth Frobenius norm approximations may be easier to obtain, but it is not clear if these can be strengthened to operator norm approximations.
For many applications operator norm approximations seem necessary, such as when these circuits appear as subroutines in larger quantum algorithms.

\subsubsection*{Beyond one qubit of control}
\label{sec:future-qutrit}
The techniques presented above appear to be specific to single-qubit controls from one unitary to the next.
If several qubits from $U^{(i-1)}$ are used to control a $U^{(i)}$, then we do not know how to obtain any asymptotic depth reduction.
This can be seen already from a staircase of \textit{qutrit}-controlled unitaries.
Attempts to extend the argument above to qutrits seem to require a $3\times 3$ version of the CS decomposition, which is false in general.
Parameter counting suggests this is not too surprising, but for completeness we include a concrete counterexample.

\begin{proposition}[Counterexample to ``$3\times 3$ CS decomposition'']
   Let $U$ be any unitary matrix proportional to a $6\times 6$ matrix containing the following entries
\[\left(\begin{array}{cc|cccc}
  1 & 0  & \phantom{1}& \phantom{1}& \phantom{1}&\phantom{1}\\
    0 & 2  & & & & \\
    1 & 2 & \multicolumn{4}{c}{\multirow{2}{*}{$*$}}\\
    0 & 3  & & & & \\
    \cline{1-2}
  \multicolumn{2}{c|}{\multirow{2}{*}{$*$}}  & \\
  \phantom{.} &
\end{array}\right)\,,
\]
Then there do \textit{not} exist block-diagonal unitaries $S$ and $T$ and $2\times 2$ diagonal matrices $\Sigma_{i,j}$, $i.j\in[3]$ such that
\begin{equation}
\label{eq:nonadiag}
    \begin{pmatrix}
    S_1 & 0 & 0\\
    0 & S_2 & 0\\
    0 & 0 & S_3
\end{pmatrix}
U
\begin{pmatrix}
    T_1 & 0 & 0\\
    0 & T_2 & 0\\
    0 & 0 & T_3
\end{pmatrix}
= \begin{pmatrix}
    \Sigma_{11} & \Sigma_{12} & \Sigma_{13}\\
    \Sigma_{21} & \Sigma_{22} & \Sigma_{23}\\
    \Sigma_{31} & \Sigma_{32} & \Sigma_{33}
\end{pmatrix}\,.
\end{equation}
\end{proposition}

\begin{proof}
Suppose we do have $S = S_1 \oplus S_2 \oplus S_3$ and $T = T_1 \oplus T_2 \oplus T_3$ achieving \eqref{eq:nonadiag}. 
Then
$$
S_1\begin{pmatrix}
    1 & 0 \\ 0 & 2
\end{pmatrix}T_1 = \Sigma_{11}
$$
is diagonal, so
$$
\Sigma_{11}^\dagger \Sigma_{11} = T_1^\dagger\begin{pmatrix}
    1 & 0 \\ 0 & 4
\end{pmatrix}T_1  =  \qquad 
\begin{pmatrix}
    1 & 0 \\ 0 & 4
\end{pmatrix} \quad \text{or} \quad
\begin{pmatrix}
    4 & 0 \\ 0 & 1
\end{pmatrix} \quad 
$$
hence $T_1$ must be either diagonal or antidiagonal.
Using \eqref{eq:nonadiag} again, \[S_2\begin{pmatrix}
    1 & 2 \\ 0 & 3
\end{pmatrix}T_1 = \Sigma_{21}\]
is diagonal. 
Hence
\[
\Sigma_{21}^\dagger\Sigma_{21}=\left(S_2\begin{pmatrix}
    1 & 2 \\ 0 & 3
\end{pmatrix}T_1\right)^\dagger S_2\begin{pmatrix}
    1 & 2 \\ 0 & 3
\end{pmatrix} T_1 
= T_1^\dagger \begin{pmatrix}
    1 & 0 \\ 2 & 3
\end{pmatrix}\begin{pmatrix}
    1 & 2 \\ 0 & 3
\end{pmatrix} T_1
= T_1^\dagger \begin{pmatrix}
    1 & 2 \\ 2 & 13
\end{pmatrix} T_1
\]
is diagonal too.
However, this matrix is full since $T_1$
is a diagonal or antidiagonal unitary.
\end{proof}

\subsubsection*{General unitaries}

The techniques of this paper also do not apply to variants of Moore--Nilsson circuits where the controlled unitary operations are replaced by arbitrary two-qubit unitaries. (See~\Cref{Fig:general_unitary_staricases}.) 

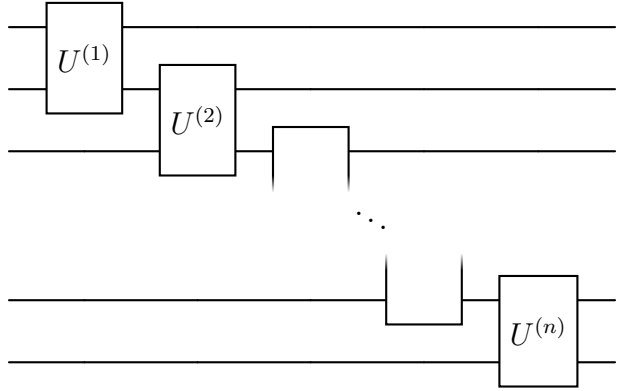
\begin{figure}
\[
\autoname{\begin{quantikz}[classical gap=1.8pt,row sep={2em,between origins},
wire types={q,q,q,q,q,q}]
     & \gate[2]{U^{(1)}} & & & & & \\
     &  & \gate[2]{U^{(2)}} & & & & \\
     & &  & \gate[2]{\phantom{U^{(3)}}} & & & \\[0.3em]
    \newwave[minimum height=3.2em]&&&\defer{\hspace{5em}\rotatebox{-30}{$\cdots$}\hspace{1em}}& \gate[2]{\phantom{U^{(3)}}} &\\[.5em]
    & & & & &\gate[2]{U^{(n)}} &  \\
    & & & & & & 
\end{quantikz}}
\]
\caption{A Moore--Nilsson circuit variant with arbitrary two-qubit unitaries. It is an open question if circuits of this form are parallelizable.}
\label{Fig:general_unitary_staricases}
\end{figure}
Standard compilation techniques let us rewrite each two-qubit unitary in this cascade as a product of single-qubit rotations and two-qubit controlled unitaries. But it is unclear if the techniques in this paper can be extended to apply to circuits of this form. 

\subsubsection*{Speedups for ``quantum dynamic programming'' in many dimensions?}

\begin{figure}
    \centering
    \begin{subfigure}{.56\textwidth}
    \centering
        \includegraphics[scale=0.6]{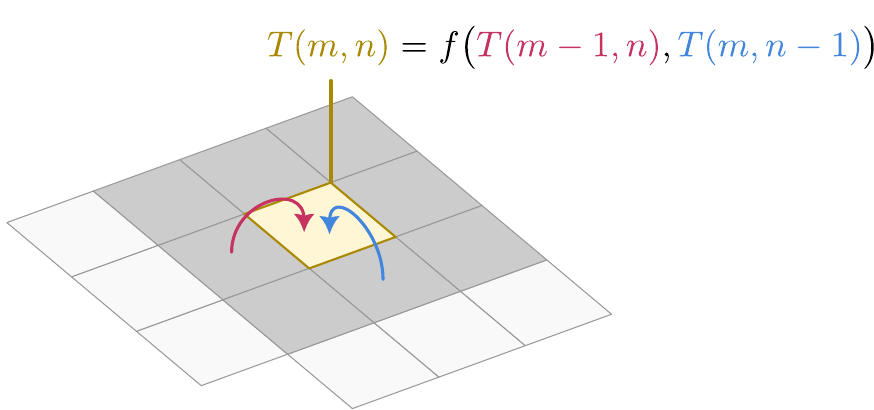}
        \vspace{.5em}
    \caption{A classical 2-dimensional memo table}
    \label{subfig:q4r-2d-c}
    \end{subfigure}
    \begin{subfigure}{.43\textwidth}
        \includegraphics[scale=0.6]{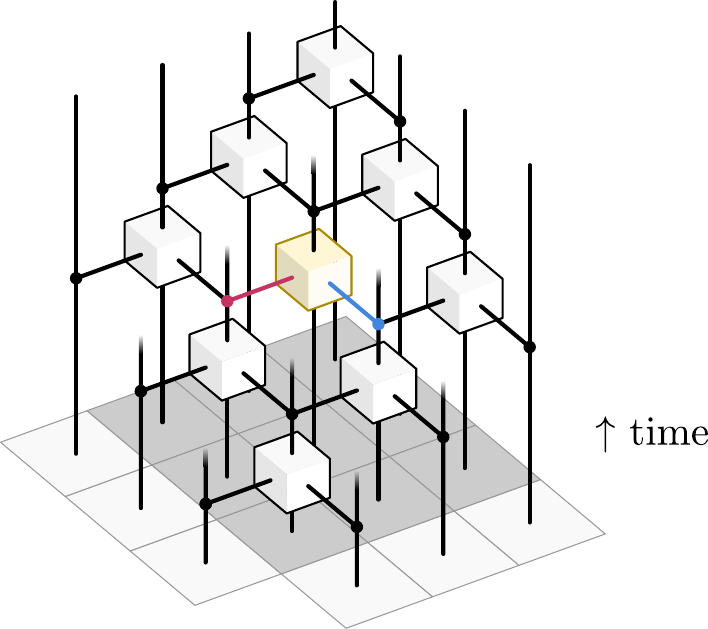}
        \vspace{-.7em}
        \caption{A quantum analogue}
        \label{subfig:q4r-2d-q}
    \end{subfigure}
    \vspace{-.5em}
    \caption{The current paper presents a quantum precomputation method for one-dimensional ``quantum memo tables.''
    Does a quantum precomputation method exist for 2D tables and up?}
    \label{fig:q4r-2d}
\end{figure}

The general parallelization techniques presented in this paper are specialized to one-dimensional ``quantum memo tables,'' or control cascade circuits.
One can envision general $k$-dimensional quantum memo tables, in analogy with classical dynamic programming---see \textit{e.g.,} \Cref{fig:q4r-2d}.
Can we parallelize such circuits?
Note that $k$-dimensional memo tables of maximum sidelength $n$---as well as their quantum analogues---can be naively computed in circuit depth $n$ (cells in the same diagonal are independent), so the appropriate goal is to asymptotically \textit{beat} depth~$n$.
Classically this is possible with the Four-Russians method \cite{arlazarov1970economical}, down to depth $kn/\log(n)$.

\subsection{Towards practical precomputation techniques}
\label{subsec:compilation}

Beyond having complexity theoretic implications, are the circuit rewriting techniques described in this paper useful in practice? Answering this question likely requires understanding how these rewriting procedures affect other resource requirements of circuits, for example their $T$ count.
And this, in turn, requires a finer-grained investigation of the procedure, particularly the $D, P$ and $R$ matrices constructed in~\Cref{lem:q4rid-Intro,lem:q4rid-moore-nilsson-Intro}.

\Cref{app:CSDecFlippedP} provides a first step in this direction, by giving inductive circuit decompositions of the $S, T,$ and $D$ unitaries generated by the CS decomposition of valley circuits. Beginning with these formula and then following the proofs outlined in this paper it should be possible to obtain clean descriptions of the circuits produced by parallelizing both Moore--Nilsson unitaries and other more-general unitary cascades of interest. This fine-grained investigation of circuits also has the potential to yield new asymptotic depth reductions for certain classes of unitary cascades, as it would allow replacing the worst-case bound of depth $4^k$ for implementing the $P$ and multiplexer $R$ gates with tighter upper bounds. 

Finally, we mention the possibility that a closer investigation of the circuits produced by our parallelization procedure might reveal the opportunity to apply even more compilation techniques. As a concrete example, we mention that the parallelization of a Moore--Nilsson cascade of identical unitaries (for example, the $\mc C (H,H,...,H)$ circuit discussed in this paper's introduction) produces columns of identical $P$ and multiplexer $R$ gates. Can the unitary ``mass-production'' theorems of~\cite{kretschmer2022quantum} be used in this setting?

\printbibliography   

\appendix

\section{Explicit formulas for the CS decomposition of a valley circuit}
\label{scc:CSDecFlippedP}
\label{app:CSDecFlippedP}


The CS decomposition 
of a valley circuit can be written in terms of the CS decomposition of
the unitaries which define the valley circuit.
In \Cref{lem:SinglePyramidDelta} we provided an algebraic 
formula for the stubs of the valley circuit.
In this section, we provide formulas for every component of a CS Decomposition for a basic valley; 
here, we do it totally in terms of circuit diagrams.

We start by reviewing CS decomposition notation. Then we state the main result  \Cref{prop:CSDecFlippedP}.

\subsection{Notation for this appendix}
We begin with a remark on notation, so that the forthcoming set of formulas is consistent with the ones which came before.
In this section, when we take the CS decomposition of a circuit $\Xi$, instead of using
\begin{align}
    \Xi = \autoname{\begin{quantikz}[wire types = {q,b}, classical gap = 2pt]
        & \coctrl{0}\wire[d]{q} & \gate{\Phi^\dagger} & \gate[2]{D^\Xi} & \gate{\Phi} & \coctrl{0}\wire[d]{q} & \\
        & \gate{S^\Xi} &&&& \gate{T^\Xi} &
    \end{quantikz}}
\end{align}
we will use a shorthanded version
\begin{align}
    \Xi = \autoname{\begin{quantikz}[wire types = {q,b}, classical gap = 2pt]
        & \coctrl{0}\wire[d]{q} & \gate[2]{\Sigma^\Xi} & \coctrl{0}\wire[d]{q} & \\
        & \gate{S^\Xi} && \gate{T^\Xi} &
    \end{quantikz}}
\end{align}
where we take
\begin{align}
    \Sigma^\Xi = \autoname{\begin{quantikz}[wire types = {q,b}, classical gap = 2pt]
        & \gate{\Phi^\dagger} & \gate[2]{D^\Xi} & \gate{\Phi} & \\
        &&&&
    \end{quantikz}} = \autoname{\begin{pmatrix}
        \Sigma_1^\Xi & \Sigma_2^\Xi \\ -\Sigma_2^\Xi & \Sigma_1^\Xi
    \end{pmatrix}}
\end{align}
where $\Sigma_1^\Xi$ and $\Sigma_2^\Xi$ are nonnegative diagonal matrices.


\subsection{The formula for the CS Decomposition of a valley}

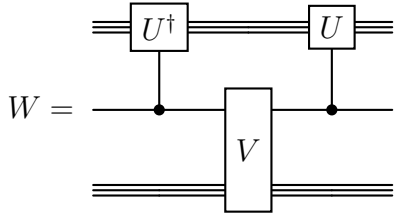
\begin{figure}[h!]
        \centering
        $$
        W = \autoname{\begin{quantikz}[classical gap=0.07cm, wire types = {b, q, b}]
            & \gate{U^\dagger} & & \gate{U} & \\
            & \ctrl{0}\wire[u]{q} & \gate[2]{V} & \ctrl{0}\wire[u]{q} &\\
            & & \ghost{V} & &
         \end{quantikz}}
        $$
        \caption{The basic valley}
        \label{fig:CorePyramid}
\end{figure}

\begin{proposition}
    \label{prop:CSDecFlippedP}
    Given a valley $W$ of $U$ and $V$ in accordance with \Cref{fig:CorePyramid} and CS Decompositions of $U$ and $V$, a CS Decomposition of $W$
    \begin{align}
        W =
        \autoname{\begin{quantikz}[wire types = {q,b,q,b}, classical gap = 2pt]
            & \coctrl{0}\wire[d]{q} & \gate[4]{\Sigma^W} & \coctrl{0}\wire[d]{q} & \\
            & \gate[3]{S^W} && \gate[3]{T^W} & \\
            &&&& \\
            &&&&
        \end{quantikz}}
    \end{align}
    has components
    \begin{align*}
        \Sigma^W := \begin{pmatrix}
            \Sigma_1^W & \Sigma_2^W \\
            -\Sigma_2^W & \Sigma_1
        \end{pmatrix}
        \qquad \text{with} \qquad 
        \Sigma_2^W := \autoname{\begin{quantikz}[wire types = {b,q,b}, classical gap = 2pt]
            & \gate{\Sigma_2^U} & \\
            && \\
            & \gate{\Sigma_2^V} &
        \end{quantikz}}
        \qquad \text{and} \qquad
        \Sigma_1^W := \sqrt{I - (\Sigma_2^W)^2}
        \\
        S^W = \autoname{\begin{quantikz}[wire types = {q,b,q,b}, classical gap = 2pt]
            & \coctrl{0}\wire[d]{q} && \octrl{2} & \\
            & \gate{\widetilde R^{U\dagger}} & \coctrl{0}\wire[d]{q} && \\
            & \coctrl{0}\wire[d]{q}\wire[u]{q} & \gate{L^\dagger} & \gate{Z} & \\
            & \gate{S^V} & \coctrl{0}\wire[u]{q} &&
        \end{quantikz}}
        \qquad \text{and} \qquad
        T^W = \autoname{\begin{quantikz}[wire types = {q,b,q,b}, classical gap = 2pt]
            & \ctrl{2} &&& \coctrl{0}\wire[d]{q} & \\
            && \coctrl{0}\wire[d]{q} && \gate{\widetilde R^U} & \\
            & \gate{Z} & \gate{L} & \gate{X} & \coctrl{0}\wire[d]{q}\wire[u]{q} & \\
            && \coctrl{0}\wire[u]{q} && \gate{T^V} &
        \end{quantikz}}
    \end{align*}
    where $\widetilde R^U$ is represented by
    \begin{align}
        \widetilde R^U = 
        \autoname{\begin{quantikz}[wire types = {q,b,q}, classical gap = 2pt]
            & \coctrl{0}\wire[d]{q} & \\
            & \gate{\widetilde R^U} & \\
            & \coctrl{0}\wire[u]{q} &
        \end{quantikz}} = 
        \autoname{\begin{quantikz}[wire types = {q,b,q}, classical gap = 2pt]
            & \coctrl{0}\wire[d]{q} & \coctrl{0}\wire[d]{q} & \\
            & \gate{S^{U\dagger}} & \gate{T^U} & \\
            & \octrl{-1} & \ctrl{-1} &
        \end{quantikz}} = 
        \autoname{\begin{quantikz}[wire types = {q,b,q}, classical gap = 2pt]
            & \octrl{1} & \ctrl{1} & \octrl{1} & \ctrl{1} & \\
            & \gate{S_0^{U\dagger}} & \gate{S_1^{U\dagger}} & \gate{T_0^{U}} & \gate{T_1^{U}} & \\
            & \octrl{-1} & \octrl{-1} & \ctrl{-1} & \ctrl{-1} &
        \end{quantikz}}
    \end{align}
    and $L$ represents a list of $(\dim W)/4$ one-qubit unitaries are defined in \Cref{eq:BW}.
\end{proposition}
\noindent
\textit{Note:} $\widetilde R$ is the same as $R$ used in the parallelization identity with reversed qubits.

The proof is based on two lemmas each presented in its own subsubsection.

\subsubsection{Reduction from a general valley to a valley of diagonal matrices}

\begin{lemma}
    Given CS decompositions of $U$ and $V$
    \begin{align}
        U = \autoname{\begin{quantikz}[wire types = {q,b}, classical gap = 2pt]
            & \coctrl{0}\wire[d]{q} & \gate[2]{\Sigma^U} & \coctrl{0}\wire[d]{q} & \\
            & \gate{S^U} & \ghost{\Sigma^U} & \gate{T^U} &
        \end{quantikz}}
        \qquad \text{and} \qquad
        V = \autoname{\begin{quantikz}[wire types = {q,b}, classical gap = 2pt]
            & \coctrl{0}\wire[d]{q} & \gate[2]{\Sigma^V} & \coctrl{0}\wire[d]{q} & \\
            & \gate{S^V} & \ghost{\Sigma^V} & \gate{T^V} &
        \end{quantikz}}.
    \end{align}
    We can factor the waterfall $W$ defined as 
    \begin{align}
        W = \autoname{\begin{quantikz}[wire types = {q,b,q,b}, classical gap = 2pt]
            & \gate[2]{U^\dagger} && \gate[2]{U} & \\
            & \ghost{U} && \ghost{U^\dagger} & \\
            & \ctrl{-1} & \gate[2]{V} & \ctrl{-1} & \\
            && \ghost{V} &&
        \end{quantikz}}
    \end{align}
    into
    \begin{align*}
        W = 
        \autoname{\begin{quantikz}[wire types = {q,b,q,b}, classical gap = 2pt]
            & \coctrl{0}\wire[d]{q} & \\
            & \gate{\widetilde R^{U\dagger}} & \\
            & \coctrl{0}\wire[u]{q}\wire[d]{q} & \\
            & \gate{S^V} &
        \end{quantikz}}\cdot
        \autoname{\begin{quantikz}[wire types = {q,b,q,b}, classical gap = 2pt]
            & \gate[2]{\Sigma^{U\dagger}} && \gate[2]{\Sigma^{U}} & \\
            & \ghost{U} && \ghost{U} & \\
            & \ctrl{-1} & \gate[2]{\Sigma^V} & \ctrl{-1} & \\
            && \ghost{V} &&
        \end{quantikz}}\cdot
        \autoname{\begin{quantikz}[wire types = {q,b,q,b}, classical gap = 2pt]
            & \coctrl{0}\wire[d]{q} & \\
            & \gate{\widetilde R^{U}} & \\
            & \coctrl{0}\wire[u]{q}\wire[d]{q} & \\
            & \gate{T^V} &
        \end{quantikz}}
    \end{align*}
    where the multiplexer $\widetilde R^{U}$ refers to the following circuit:
    \begin{align}
        \autoname{\begin{quantikz}[wire types = {q,b,q}, classical gap = 2pt]
            & \coctrl{0}\wire[d]{q} & \\
            & \gate{\widetilde R^U} & \\
            & \coctrl{0}\wire[u]{q} &
        \end{quantikz}} = 
        \autoname{\begin{quantikz}[wire types = {q,b,q}, classical gap = 2pt]
            & \coctrl{0}\wire[d]{q} & \coctrl{0}\wire[d]{q} & \\
            & \gate{S^{U\dagger}} & \gate{T^U} & \\
            & \octrl{-1} & \ctrl{-1} &
        \end{quantikz}} = 
        \autoname{\begin{quantikz}[wire types = {q,b,q}, classical gap = 2pt]
            & \octrl{1} & \ctrl{1} & \octrl{1} & \ctrl{1} & \\
            & \gate{S_0^{U\dagger}} & \gate{S_1^{U\dagger}} & \gate{T_0^{U}} & \gate{T_1^{U}} & \\
            & \octrl{-1} & \octrl{-1} & \ctrl{-1} & \ctrl{-1} &
        \end{quantikz}}
    \end{align}
\end{lemma}

\begin{proof}
    We start by substituting $U$ and $V$ with their CS decompositions
    $$
    W = 
    \autoname{\begin{quantikz}[wire types = {q,b,q,b}, classical gap = 2pt]
        & \coctrl{0}\wire[d]{q} & \gate[2]{\Sigma^{U\dagger}} & \coctrl{0}\wire[d]{q} &&&& \coctrl{0}\wire[d]{q} & \gate[2]{\Sigma^{U}} & \coctrl{0}\wire[d]{q} & \\
        & \gate{T^{U\dagger}} & \ghost{\Sigma^U} & \gate{S^{U\dagger}} &&&& \gate{S^{U}} & \ghost{\Sigma^U} & \gate{T^{U}} & \\
        & \ctrl{0}\wire[u]{q} & \ctrl{0}\wire[u]{q} & \ctrl{0}\wire[u]{q} & \coctrl{0}\wire[d]{q} & \gate[2]{\Sigma^V} & \coctrl{0}\wire[d]{q} & \ctrl{0}\wire[u]{q} & \ctrl{0}\wire[u]{q}& \ctrl{0}\wire[u]{q} & \\    
        &&&& \gate{S^V} & \ghost{\Sigma^V} & \gate{T^V} &&&&
    \end{quantikz}}
    $$
    We can multiply by 
    $\autoname{\begin{quantikz}[wire types = {q,b,q,b}, classical gap = 2pt]
        & \coctrl{0}\wire[d]{q} & \coctrl{0}\wire[d]{q} & \\
        & \gate{S^{U}} & \gate{S^{U\dagger}} & \\
        & \octrl{-1} & \octrl{-1} &
    \end{quantikz}}$
    which is the identity on both the left and right, 
    and then commute them through the closed controls to get
    \begin{align}
        W = \autoname{\begin{quantikz}[wire types = {q,b,q,b}, classical gap = 2pt]
            & \coctrl{0}\wire[d]{q} & \coctrl{0}\wire[d]{q} & \gate[2]{\Sigma^{U\dagger}} & \coctrl{0}\wire[d]{q} &&&& \coctrl{0}\wire[d]{q} & \gate[2]{\Sigma^{U}} & \coctrl{0}\wire[d]{q} & \coctrl{0}\wire[d]{q} & \\
            & \gate{S^U} & \gate{T^{U\dagger}} & \ghost{\Sigma^U} & \gate{S^{U\dagger}} &&&& \gate{S^{U}} & \ghost{\Sigma^U} & \gate{S^{U\dagger}} & \gate{T^U} & \\
            & \octrl{0}\wire[u]{q} & \ctrl{-1} & \ctrl{0}\wire[u]{q} && \coctrl{0}\wire[d]{q} & \gate[2]{\Sigma^V} & \coctrl{0}\wire[d]{q} && \ctrl{0}\wire[u]{q} & \octrl{0}\wire[u]{q} & \ctrl{-1} & \\
            &&&&& \gate{S^V} & \ghost{\Sigma^V} & \gate{T^V} &&&&&
        \end{quantikz}}
    \end{align}
    The center $S^U$ terms cancel, the $S^V$ and $T^V$ terms commute out, and the $S^U$ and $T^U$ terms combine into $\widetilde R^U$ to get
    \begin{align}
        W = \autoname{\begin{quantikz}[wire types = {q,b,q,b}, classical gap = 2pt]
            & \coctrl{0}\wire[d]{q} & \gate[2]{\Sigma^{U\dagger}} && \gate[2]{\Sigma^{U}} & \coctrl{0}\wire[d]{q} & \\
            & \gate{\widetilde R^{U\dagger}} & \ghost{\Sigma^U} && \ghost{\Sigma^U} & \gate{\widetilde R^U} & \\
            & \coctrl{0}\wire[u]{q}\wire[d]{q} & \ctrl{-1} & \gate[2]{\Sigma^V} & \ctrl{-1} & \coctrl{0}\wire[u]{q}\wire[d]{q} & \\
            & \gate{S^V} && \ghost{\Sigma^V} && \gate{T^V} &
        \end{quantikz}}
    \end{align}
    which confirms the factorization.
\end{proof}


\subsubsection{Valleys of diagonal matrices}

\begin{lemma}
    Let $\Sigma^U = \begin{pmatrix}
        \Sigma_1^U & \Sigma_2^U \\ -\Sigma_2^U & \Sigma_1^U
    \end{pmatrix}$ be a unitary circuit where $\Sigma_1^U$ and $\Sigma_2^U$ are diagonal and let $\Sigma^V$ be similarly defined.
    Given a valley circuit $W^\Sigma$ of $\Sigma^U$ and $\Sigma^V$ in accordance with

    \begin{align}
        W^\Sigma = \autoname{\begin{quantikz}[wire types = {q,b,q,b}, classical gap = 2pt]
            & \gate[2]{\Sigma^{U\dagger}} && \gate[2]{\Sigma^{U}} & \\
            & \ghost{\Sigma^U} && \ghost{\Sigma^U} & \\
            & \ctrl{-1} & \gate[2]{\Sigma^V} & \ctrl{-1} & \\
            && \ghost{\Sigma^V} &&
        \end{quantikz}},
    \end{align}
    a CS Decomposition of $W^\Sigma$ is given by
    \begin{align*}
        \Sigma_2^{W^\Sigma} := \autoname{\begin{quantikz}[wire types = {b,q,b}, classical gap = 2pt]
            & \gate{\Sigma_2^U} & \\
            && \\
            & \gate{\Sigma_2^V} &
        \end{quantikz}}
        \qquad \text{and} \qquad
        \Sigma_1^{W^\Sigma} := \sqrt{I - (\Sigma_2^{W^\Sigma})^2}
        \\
        S^{W^\Sigma} = \autoname{\begin{quantikz}[wire types = {q,b,q,b}, classical gap = 2pt]
            && \octrl{2} & \\
            & \coctrl{0}\wire[d]{q} && \\
            & \gate{L^\dagger} & \gate{Z} & \\
            & \coctrl{0}\wire[u]{q} &&
        \end{quantikz}}
        \qquad \text{and} \qquad
        T^{W^\Sigma} = \autoname{\begin{quantikz}[wire types = {q,b,q,b}, classical gap = 2pt]
            & \ctrl{2} &&& \\
            && \coctrl{0}\wire[d]{q} && \\
            & \gate{Z} & \gate{L} & \gate{X} & \\
            && \coctrl{0}\wire[u]{q} &&
        \end{quantikz}}
    \end{align*}
    where $L$ is a list of $\dim W/4$ one-qubit unitaries.
\end{lemma}
\begin{proof}
    We write $W^{\Sigma}$ as a block  $ 2 \times 2$ matrix, based on inputs and outputs on the middle wire.
    This allows us to write $W^\Sigma$ as the sum of the $4$ blocks. We will use $\ketbra{input}{output}$ on the middle wire to denote which block is selected. The sum of the four blocks can be written:  
    \begin{align*}
        W^\Sigma = 
        \autoname{\begin{quantikz}[wire types = {q,b,q,b}, classical gap = 2pt]
            & \ghost{\Sigma^U} & \\
            & \ghost{\Sigma^U} & \\
            & \gate{\ketbra00} & \\
            & \gate{\Sigma_1^V} &
        \end{quantikz}}
        \ + \
        \autoname{\begin{quantikz}[wire types = {q,b,q,b}, classical gap = 2pt]
            & \ghost{\Sigma^U} & \\
            & \ghost{\Sigma^U} & \\
            & \gate{\ketbra11} & \\
            & \gate{\Sigma_1^V} &
        \end{quantikz}}
        \ + \
        \autoname{\begin{quantikz}[wire types = {q,b,q,b}, classical gap = 2pt]
            & \gate[2]{\Sigma^U} & \\
            & \ghost{\Sigma^U} & \\
            & \gate{\ketbra01} & \\
            & \gate{\Sigma_2^V} &
        \end{quantikz}}
        \ + \
        \autoname{\begin{quantikz}[wire types = {q,b,q,b}, classical gap = 2pt]
            & \gate[2]{\Sigma^{U\dagger}} & \\
            & \ghost{\Sigma^U} & \\
            & \gate{\ketbra10} & \\
            & \gate{-\Sigma_2^V} &
        \end{quantikz}}
    \end{align*}
    We next consider the blocks of $W^\Sigma$ individually. 
        First, we start by analyzing the off-diagonal blocks.

\bs 

\noindent
    {\bf Off-diagonal blocks.}
    $W^\Sigma_{01}$ is the off-diagonal block where we take the input $\bra0$ and the output $\ket1$ on the top wire. With the identity on the top wire, the first two terms of $W^\Sigma_{01}$ are $0$ and we are left with
    \begin{align*}
        W^\Sigma_{01} = \autoname{\begin{quantikz}[wire types = {q,b,q,b}, classical gap = 2pt]
            \bra0 & \gate[4]{W^\Sigma} & \ket{1} \\
            & \ghost{W} & \\
            & \ghost{W} & \\
            & \ghost{W} &
        \end{quantikz}}
        =
        \autoname{\begin{quantikz}[wire types = {b,q,b}, classical gap = 2pt]
            & \gate{\Sigma_2^U} & \\
            & \gate{\ketbra01} & \\
            & \gate{\Sigma_2^V} &
        \end{quantikz}} + 
        \autoname{\begin{quantikz}[wire types = {b,q,b}, classical gap = 2pt]
            & \gate{-\Sigma_2^U} & \\
            & \gate{\ketbra10} & \\
            & \gate{-\Sigma_2^V} &
        \end{quantikz}}
        =
        \autoname{\begin{quantikz}[wire types = {b,q,b}, classical gap = 2pt]
            & \gate{\Sigma_2^U} & \\
            & \gate{X} & \\
            & \gate{\Sigma_2^V} &
        \end{quantikz}}
        =
        \autoname{\begin{quantikz}[wire types = {b,q,b}, classical gap = 2pt]
            & \gate[3]{\Sigma_2^{W^\Sigma}} && \\
            && \gate{X} & \\
            &&&
        \end{quantikz}}
    \end{align*}
    because $\ketbra01 + \ketbra10 = X$.
    
    We can perform similar operations to find $W_{10}^\Sigma$
    \begin{align*}
        W_{10}^\Sigma = 
        \autoname{\begin{quantikz}[wire types = {q,b,q,b}, classical gap = 2pt]
            \bra1 & \gate[4]{W^\Sigma} & \ket{0} \\
            & \ghost{W} & \\
            & \ghost{W} & \\
            & \ghost{W} &
        \end{quantikz}}
        =
        \autoname{\begin{quantikz}[wire types = {b,q,b}, classical gap = 2pt]
            & \gate{-\Sigma_2^U} & \\
            & \gate{\ketbra01} & \\
            & \gate{\Sigma_2^V} &
        \end{quantikz}} + 
        \autoname{\begin{quantikz}[wire types = {b,q,b}, classical gap = 2pt]
            & \gate{\Sigma_2^U} & \\
            & \gate{\ketbra10} & \\
            & \gate{-\Sigma_2^V} &
        \end{quantikz}}
        = 
        \autoname{\begin{quantikz}[wire types = {b,q,b}, classical gap = 2pt]
            & \gate{\Sigma_2^U} & \\
            & \gate{X} & \\
            & \gate{-\Sigma_2^V} &
        \end{quantikz}}
        =
        \autoname{\begin{quantikz}[wire types = {b,q,b}, classical gap = 2pt]
            & \gate[3]{-\Sigma_2^{W^\Sigma}} && \\
            && \gate{X} & \\
            &&&
        \end{quantikz}}
    \end{align*}

    
    \noindent
    {\bf Diagonal blocks.}
    Next, we consider the $W_{00}^\Sigma$ block. We have
    \begin{align}
        W_{00}^\Sigma = \autoname{\begin{quantikz}[wire types = {b,q,b}, classical gap = 2pt]
            & \ghost{\Sigma^U} & \\
            & \gate{\ketbra00} & \\
            & \gate{\Sigma_1^V} & 
        \end{quantikz}}
        +
        \autoname{\begin{quantikz}[wire types = {b,q,b}, classical gap = 2pt]
            & \ghost{\Sigma^U} & \\
            & \gate{\ketbra11} & \\
            & \gate{\Sigma_1^V} & 
        \end{quantikz}}
        + 
        \autoname{\begin{quantikz}[wire types = {b,q,b}, classical gap = 2pt]
            & \gate{\Sigma_1^U} & \\
            & \gate{\ketbra01} & \\
            & \gate{\Sigma_2^V} & 
        \end{quantikz}}
        +
        \autoname{\begin{quantikz}[wire types = {b,q,b}, classical gap = 2pt]
            & \gate{\Sigma_1^U} & \\
            & \gate{\ketbra10} & \\
            & \gate{-\Sigma_2^V} & 
        \end{quantikz}}
    \end{align}
    We can combine the first two terms as $\ketbra00 + \ketbra11 = I_2$ and we can combine the second two terms as $\ketbra01 - \ketbra10 = ZX$.
    $W_{00}^\Sigma$ can then be written as
    \begin{align}
        W_{00}^\Sigma = \autoname{\begin{quantikz}[wire types = {b,q,b}, classical gap = 2pt]
            & \ghost{\Sigma} & \\
            & \ghost{X} & \\
            & \gate{\Sigma_1^V} &
        \end{quantikz}} + \autoname{\begin{quantikz}[wire types = {b,q,b}, classical gap = 2pt]
            & \gate{\Sigma_1^U} & \\
            & \gate{ZX} & \\
            & \gate{\Sigma_2^V} &
        \end{quantikz}}
    \end{align}
    We observe that $W_{00}^\Sigma$ is diagonal except on the middle wire. If we were to swap the middle and bottom wires using $P_{swap}$, we would have
    \begin{align}
        P_{swap}^\dagger W_{00}^\Sigma P_{swap} = \autoname{\begin{quantikz}[wire types = {b,b,q}, classical gap = 2pt]
            & \ghost{\Sigma} & \\
            & \gate{\Sigma_1^V} & \\
            & \ghost{X} &
        \end{quantikz}} + \autoname{\begin{quantikz}[wire types = {b,b,q}, classical gap = 2pt]
            & \gate{\Sigma_1^U} & \\
            & \gate{\Sigma_2^V} & \\
            & \gate{ZX} &
        \end{quantikz}}.
    \end{align}
    The matrix representing this circuit is the sum of two matrices. 
    The first circuit is $I \otimes \Sigma_1^V \otimes  I$ which is block diagonal with the 
    blocks of the form $a_j I_2$ with $a_j$ a real number.
    The second circuit is $\Sigma_1^U \otimes \Sigma_2^V \otimes \begin{pmatrix}
        0 & 1 \\ -1 & 0
    \end{pmatrix}$
    which is block diagonal with the 
    blocks of the form $b_j 
    \begin{pmatrix}
        0 & 1 \\ -1 & 0
    \end{pmatrix}$  with $b_j$ a real number.

Combining 
    these, $P_{swap}^\dagger W_{00}^\Sigma P_{swap}$ is a block diagonal matrix with blocks
    \begin{align}
    \label{B_j}
        B_j = \begin{pmatrix}
            a_j & b_j \\ -b_j & a_j
        \end{pmatrix}
    \end{align}
    With this, we can replace the top sets of wires with a large set of controls:
    \begin{align}
        P_{swap}^\dagger W_{00}^\Sigma P_{swap} = \autoname{\begin{quantikz}[wire types = {b,b,q}, classical gap = 2pt]
            & \coctrl{0}\wire[d]{q} & \\
            & \coctrl{0}\wire[d]{q} & \\
            & \gate{B} &
        \end{quantikz}}.
    \end{align}
    Reversing the permutations gives us the result that we can write $W_{00}^\Sigma$ as a doubly multi-controlled list of one-qubit unitaries:
    \begin{align}
        W_{00}^\Sigma = \autoname{\begin{quantikz}[wire types = {b,q,b}, classical gap = 2pt]
            & \coctrl{0}\wire[d]{q} & \\
            & \gate{B} & \\
            & \coctrl{0}\wire[u]{q} &
        \end{quantikz}}.
    \end{align}

    We then apply the $X$ gate that converts $W_{01}^\Sigma$ and $W_{10}^\Sigma$ into $\Sigma_2^W$, and then define $B^W$ as
    \begin{align}
        \autoname{\begin{quantikz}[wire types = {b,q,b}, classical gap = 2pt]
            & \coctrl{0}\wire[d]{q} & \\
            & \gate{B^W} & \\
            & \coctrl{0}\wire[u]{q} &
        \end{quantikz}}
        :=
        \autoname{\begin{quantikz}[wire types = {b,q,b}, classical gap = 2pt]
            & \coctrl{0}\wire[d]{q} && \\
            & \gate{B} & \gate{X} & \\
            & \coctrl{0}\wire[u]{q} &&
        \end{quantikz}}
        =
        \autoname{\begin{quantikz}[wire types = {b,q,b}, classical gap = 2pt]
            & \gate[3]{W_{00}^\Sigma} && \\
            && \gate{X} & \\
            &&&
        \end{quantikz}}
        =
        \autoname{\begin{quantikz}[wire types = {b,q,b}, classical gap = 2pt]
            & \ghost{\Sigma} & \\
            & \gate{X} & \\
            & \gate{\Sigma_1^V} &
        \end{quantikz}} + \autoname{\begin{quantikz}[wire types = {b,q,b}, classical gap = 2pt]
            & \gate{\Sigma_1^U} & \\
            & \gate{Z} & \\
            & \gate{\Sigma_2^V} &
        \end{quantikz}}
    \end{align}
where   
$B^W:= (B^W_1,B^W_2,..., B^W_{\dim(W)/4})$ with 
$$B^W_j = \begin{pmatrix}
    b_j & a_j \\ a_j & -b_j
\end{pmatrix}.$$
The matrix $B^W_j$ has eigenvalues are $c_j$ and $-c_j$ where $c_j = \sqrt{a_j^2 + b_j^2}$.
We let 
$L_j$ denote a unitary matrix diagonalizing $B^W_j$, namely, 
     
    \begin{align}
        B^W_j = L_j^\dagger \begin{pmatrix}
            c_j & 0 \\ 0 & -c_j
        \end{pmatrix} L_j = L_j^\dagger c_jZ L_j
    \end{align}
    Form the list $L =(L_1, ..., L_{\dim(W)/4})$ and use it in the circuit
    \begin{align}
    \label{eq:BW}
        \autoname{\begin{quantikz}[wire types = {b,q,b}, classical gap = 2pt]
            & \coctrl{0}\wire[d]{q} & \\
            & \gate{B^W}& \\
            & \coctrl{0}\wire[u]{q} &
        \end{quantikz}}
        =
        \autoname{\begin{quantikz}[wire types = {b,q,b}, classical gap = 2pt]
            & \coctrl{0}\wire[d]{q} & \coctrl{0}\wire[d]{q} && \coctrl{0}\wire[d]{q} & \\
            & \gate{L^\dagger} & \gate{C} & \gate{Z} & \gate{L} & \\
            & \coctrl{0}\wire[u]{q} & \coctrl{0}\wire[u]{q} && \coctrl{0}\wire[u]{q} &
        \end{quantikz}}
    \end{align}
    where $C := (c_1I, c_2I, ..., c_{\dim(W)/4}I)$. Notably, $C$ ends up being a real diagonal matrix with entries equal to the singular values of $W_{00}^\Sigma$. Since $W$ is unitary and $W_{01}^\Sigma$ times $X$ on the middle wire is equal to $\Sigma_2^{W^\Sigma}$ and $\Sigma_1^{W^\Sigma} = \sqrt{1-(\Sigma_2^{W^\Sigma})^2}$, this means that $C$ has the same singular values as $W_{00}^\Sigma$ in the same order as $\Sigma_1^{W^\Sigma}$, and
    \begin{align}
        \autoname{\begin{quantikz}[wire types = {b,q,b}, classical gap = 2pt]
            & \coctrl{0}\wire[d]{q} & \\
            & \gate{C}& \\
            & \coctrl{0}\wire[u]{q} &
        \end{quantikz}}
        =
        \Sigma_1^{W^\Sigma}
    \end{align}
    so
    \begin{align}
        \autoname{\begin{quantikz}[wire types = {b,q,b}, classical gap = 2pt]
            & \coctrl{0}\wire[d]{q} & \\
            & \gate{B^W}& \\
            & \coctrl{0}\wire[u]{q} &
        \end{quantikz}}
        =
        \autoname{\begin{quantikz}[wire types = {b,q,b}, classical gap = 2pt]
            & \coctrl{0}\wire[d]{q} & \gate[3]{\Sigma_1^{W^\Sigma}} && \coctrl{0}\wire[d]{q} & \\
            & \gate{L^\dagger} && \gate{Z} & \gate{L} & \\
            & \coctrl{0}\wire[u]{q} &&& \coctrl{0}\wire[u]{q} &
        \end{quantikz}}
    \end{align}

    The computations above give us the recipes needed to evaluate the CS Decomposition stated in the lemma to 
    verify that it is indeed $W^\Sigma$. We present this next.

\bs

\noindent 
    {\bf Evaluating the CS decomposition.}
    The CS decomposition expands to
    \begin{align}
        \autoname{\begin{quantikz}[wire types = {q,b,q,b}, classical gap = 2pt]
            && \octrl{2} & \gate[4]{\Sigma^{W^\Sigma}} & \ctrl{2} &&& \\
            & \coctrl{0}\wire[d]{q} && \ghost{\Sigma_2} && \coctrl{0}\wire[d]{q} && \\
            & \gate{L^\dagger} & \gate{Z} & \ghost{\Sigma} & \gate{Z} & \gate{L} & \gate{X} & \\
            & \coctrl{0}\wire[u]{q} && \ghost{\Sigma} && \coctrl{0}\wire[u]{q} &&
        \end{quantikz}}
    \end{align}
    We write the CS decomposition as a block $2 \times 2$ matrix, based on inputs and output of the top wire to get
    \begin{align}
        = \autoname{\begin{quantikz}[wire types = {q,b,q,b}, classical gap = 2pt]
            &&& \gate{\ketbra00} &&& \\
            & \coctrl{0}\wire[d]{q} && \gate[3]{\Sigma_1^{W^\Sigma}} & \coctrl{0}\wire[d]{q} && \\
            & \gate{L^\dagger} & \gate{Z} & \ghost{\Sigma} & \gate{L} & \gate{X} & \\
            & \coctrl{0}\wire[u]{q} && \ghost{\Sigma} & \coctrl{0}\wire[u]{q} &&
        \end{quantikz}}
        \ + \
        \autoname{\begin{quantikz}[wire types = {q,b,q,b}, classical gap = 2pt]
            && \gate{\ketbra11} &&&& \\
            & \coctrl{0}\wire[d]{q} & \gate[3]{\Sigma_1^{W^\Sigma}} && \coctrl{0}\wire[d]{q} && \\
            & \gate{L^\dagger} & \ghost{\Sigma} & \gate{Z} & \gate{L} & \gate{X} & \\
            & \coctrl{0}\wire[u]{q} & \ghost{\Sigma} && \coctrl{0}\wire[u]{q} &&
        \end{quantikz}} \\
        \ + \ 
        \autoname{\begin{quantikz}[wire types = {q,b,q,b}, classical gap = 2pt]
            && \gate{\ketbra10} &&& \\
            & \coctrl{0}\wire[d]{q} & \gate[3]{-\Sigma_2^{W^\Sigma}} & \coctrl{0}\wire[d]{q} && \\
            & \gate{L^\dagger} & \ghost{\Sigma} & \gate{L} & \gate{X} & \\
            & \coctrl{0}\wire[u]{q} & \ghost{\Sigma} & \coctrl{0}\wire[u]{q} &&
        \end{quantikz}}
        \ + \ 
        \autoname{\begin{quantikz}[wire types = {q,b,q,b}, classical gap = 2pt]
            &&& \gate{\ketbra01} &&&& \\
            & \coctrl{0}\wire[d]{q} && \gate[3]{\Sigma_2^{W^\Sigma}} && \coctrl{0}\wire[d]{q} && \\
            & \gate{L^\dagger} & \gate{Z} & \ghost{\Sigma} & \gate{Z} & \gate{L} & \gate{X} & \\
            & \coctrl{0}\wire[u]{q} && \ghost{\Sigma} && \coctrl{0}\wire[u]{q} &&
        \end{quantikz}}
    \end{align}
    For the first two circuits, $Z$ commutes with $\Sigma_1^{W^\Sigma}$ so the $\ketbra00$ and $\ketbra11$ blocks are equal and
    \begin{align}
        \autoname{\begin{quantikz}[wire types = {b,q,b}, classical gap = 2pt]
            & \coctrl{0}\wire[d]{q} & \gate[3]{\Sigma_1^{W^\Sigma}} && \coctrl{0}\wire[d]{q} && \\
            & \gate{L^\dagger} & \ghost{\Sigma_1} & \gate{Z} & \gate{L} & \gate{X} & \\
            & \coctrl{0}\wire[u]{q} & \ghost{\Sigma_1} && \coctrl{0}\wire[u]{q} &&
        \end{quantikz}}
        =
        \autoname{\begin{quantikz}[wire types = {b,q,b}, classical gap = 2pt]
            & \coctrl{0}\wire[d]{q} & \coctrl{0}\wire[d]{q} && \coctrl{0}\wire[d]{q} && \\
            & \gate{L^\dagger} & \gate{C} & \gate{Z} & \gate{L} & \gate{X} & \\
            & \coctrl{0}\wire[u]{q} & \coctrl{0}\wire[u]{q} && \coctrl{0}\wire[u]{q} &&
        \end{quantikz}}
        =
        W_{00}^\Sigma = W_{11}^\Sigma
    \end{align}
    For the other two circuits, conjugating a diagonal matrix by a Pauli Z matrix does nothing, so we evaluate the $\ketbra01$ component to
    \begin{align}
        \autoname{\begin{quantikz}[wire types = {b,q,b}, classical gap = 2pt]
            & \coctrl{0}\wire[d]{q} & \gate[3]{\Sigma_2^{W^\Sigma}} && \coctrl{0}\wire[d]{q} && \\
            & \gate{L^\dagger} & \ghost{\Sigma_1} & \gate{L} & \gate{X} & \\
            & \coctrl{0}\wire[u]{q} & \ghost{\Sigma_1} && \coctrl{0}\wire[u]{q} &&
        \end{quantikz}}
        =
        \autoname{\begin{quantikz}[wire types = {b,q,b}, classical gap = 2pt]
            & \coctrl{0}\wire[d]{q} & \gate{\Sigma_2^U} & \coctrl{0}\wire[d]{q} && \\
            & \gate{L^\dagger} && \gate{L} & \gate{X} & \\
            & \coctrl{0}\wire[u]{q} & \gate{\Sigma_2^V} & \coctrl{0}\wire[u]{q} &&
        \end{quantikz}}
        =
        \autoname{\begin{quantikz}[wire types = {b,q,b}, classical gap = 2pt]
            & \gate{\Sigma_2^U} && \\
            && \gate{X} & \\
            & \gate{\Sigma_2^V} &&
        \end{quantikz}}
        = W_{01}^\Sigma.
    \end{align}
    The $\ketbra 10$ component is the same with $-\Sigma_2^{W^\Sigma}$ so the $\ketbra10$ block evaluates to $W_{10}^\Sigma$.

        Combine the calculations above to see that the CS Decomposition evaluates to
    \begin{align}
        \ketbra00 \otimes W_{00}^\Sigma + \ketbra11 \otimes W_{11}^\Sigma + \ketbra01 \otimes W_{01}^\Sigma + \ketbra10 \otimes W_{10}^\Sigma = W^\Sigma
    \end{align} 
    which is what the lemma asserted.
\end{proof}

\end{document}